\documentclass[a4paper, 11pt]{article}

\usepackage[a4paper, margin=1in]{geometry}

\usepackage[doublespacing]{setspace}
\usepackage{amsmath} 
\usepackage{amssymb}
\usepackage{amsthm}
\usepackage{appendix} 
\usepackage{bm} 
\usepackage{enumitem} 
\usepackage{graphicx} 
\usepackage{adjustbox}
\usepackage{calc}
\usepackage{float} 
\usepackage{natbib}
\setcitestyle{maxcitenames=3}

\graphicspath{ {figures/} }
\usepackage{xcolor} 
\usepackage[colorlinks, urlcolor = magenta]{hyperref} 
\usepackage{tikz}
\usetikzlibrary{patterns, positioning, calc, arrows, automata}
\usepackage{pgfplots}
\pgfplotsset{compat=1.18}

\usepackage{float}

\usepackage{mathtools}
\usepackage{blindtext}
\usepackage{multicol}
\usepackage{etoolbox}
\usepackage{subcaption}
\usepackage{parskip}

\usepackage{xurl} 

\usepackage{siunitx}
\usepackage{tabularx} 
\usepackage{array}    
\usepackage{booktabs} 
\usepackage{mathtools} 

\usepackage{caption}
\usepackage{subcaption} 

\usepackage{lipsum}   
\usepackage{fancyhdr} 

\definecolor{CiteColor}{RGB}{0,54,150} 
\definecolor{UoGColor}{RGB}{0,54,250} 

\definecolor{unstableRed}{RGB}{255, 235, 235}
\definecolor{maybeYellow}{RGB}{255, 248, 220}
\definecolor{stableGreen}{RGB}{235, 255, 235}
\definecolor{lineRed}{RGB}{200, 50, 50}
\definecolor{lineGreen}{RGB}{50, 150, 50}

\hypersetup{
    linkcolor=UoGColor, 
    citecolor=UoGColor, 
}

\newcommand{\fignote}[1]{\caption*{\footnotesize{\textit{Note:} #1}}}

\newtheorem{proposition}{Proposition}

\newtheorem{lemma}{Lemma}

\newtheorem{corollary}{Corollary}

\theoremstyle{definition}
\newtheorem{definition}{Definition}

\newlength{\equationwidth}
\newcommand{\scaledequation}[1]{%
  \settowidth{\equationwidth}{$\displaystyle #1$}%
  \ifdim\equationwidth>\textwidth
    \resizebox{\textwidth}{!}{$\displaystyle #1$}%
  \else
    #1%
  \fi
}

\newcommand{\str}[1]{{#1}^{*}}

\title{The R\&D Productivity Puzzle: \\Innovation Networks with Heterogeneous Firms
}
\author{
    M. Sadra Heydari\thanks{Adam Smith Business School, University of Glasgow, \url{Sadra.HeydariMostafaabadi@glasgow.ac.uk}}
        \and
    Zafer Kanik\thanks{Adam Smith Business School, University of Glasgow, \url{Zafer.Kanik@glasgow.ac.uk}}
        \and
    Santiago Montoya-Blandón\thanks{ Adam Smith Business School, University of Glasgow, \url{Santiago.Montoya-Blandon@glasgow.ac.uk}}
}

\date{
    Draft: \today
\\
\href
{https://arxiv.org/abs/2512.23337}
{[\underline{Latest Version}]}
}

\begin{document}

\maketitle
\thispagestyle{empty}

\begin{abstract}
    \begin{singlespace}
\noindent 
We introduce heterogeneous R\&D productivities into an endogenous R\&D network formation model, generalizing the framework of \cite{goyal2001r}. Heterogeneous productivities endogenously create asymmetric gains from collaboration: less productive firms benefit disproportionately from links, while more productive firms exert greater R\&D effort and incur higher costs. When productivity gaps are sufficiently large, more productive firms experience lower profits from collaborating with less productive partners. As a result, the complete network—stable under homogeneity—becomes unstable, and the positive assortative (PA) network, in which firms cluster by R\&D productivity, emerges as pairwise stable. Using simulations, we show that the clustered structure delivers higher welfare than the complete network; nevertheless, welfare under this formation follows an inverted U-shape as the fraction of high-productivity firms increases, reflecting crowding-out effects at high fractions. Altogether, we uncover an R\&D productivity puzzle: economies with higher average R\&D productivity may exhibit lower welfare through (i) the formation of alternative stable networks, or (ii) a crowding-out effect of high-productivity firms. Our findings show that productivity gaps shape the organization of innovation by altering equilibrium R\&D alliances and effort. Productivity-enhancing policies must therefore account for these endogenous responses, as they may reverse intended welfare gains.
\end{singlespace}
\end{abstract}

\newpage

\section{Introduction}
 \label{sec:introduction}

Innovation is widely recognized as a fundamental driver of economic growth in modern societies \citep{romer1990endogenous, jones1995r}. A key engine of the innovation process is private firms’ investment in R\&D \citep{grossman1991quality, aghion1992constructive}, a substantial share of which takes place through collaborative R\&D alliances, as documented empirically by \citet{cassiman2002r, belderbos2004cooperative, calero2007research}.

In this paper, we study R\&D alliances and their implications for firms' profits and aggregate welfare in an R\&D network formation model with heterogeneous R\&D productivities. Following the canonical work by \citet{goyal2001r}, a growing theoretical literature, including \citet{goyal2003networks, konig2019r, hsieh2025endogenous}, studies R\&D networks in which firms raise the effectiveness of R\&D through collaboration while continuing to compete in product markets. A common benchmark in this literature abstracts from cross-firm differences in the ability to translate R\&D effort into cost-reducing innovations. We extend this framework by allowing firms’ innovation efficiencies to differ, consistent with theory and evidence that capabilities vary substantially across firms \citep[][]{magerman2016heterogeneous, chen2023capability, milan2024learningS, acemoglu2025macroeconomics}.\footnote{Further empirical studies show that the return on R\&D is driven by unobserved heterogeneity in innovation efficiency \citep{lentz2008empirical} and organizational competence \citep{henderson1994measuring}, rather than spending intensity alone. Indeed, the output elasticity of R\&D is highly firm-specific \citep{knott2008r} and depends on the interplay between a firm’s size and innovation strategy \citep{akcigit2018growth}.} We show that introducing heterogeneity generates a ``productivity puzzle": economies with larger overall R\&D productivity can exhibit lower welfare.

This paper contributes to three strands of the literature on R\&D networks. First, it extends models of endogenous R\&D network formation by incorporating firm-level heterogeneity in R\&D productivities. Second, this methodological approach demonstrates how heterogeneity leads to endogenous asymmetric benefits for connecting firms, which alter firms' incentives to collaborate and consequently generate a richer set of equilibrium network structures compared to the homogeneous benchmark. Third, we introduce and explain the mechanisms underlying the counterintuitive R\&D productivity puzzle: economies with higher average R\&D productivity may exhibit lower aggregate welfare through (i) the formation of denser, stable R\&D networks with lower individual efforts or (ii) a crowding-out effect of high-productivity firms. Welfare losses arise since the connectivity effect dominates the productivity effect in determining welfare.

Our model generalizes the framework introduced in \cite{goyal2001r} by allowing for heterogeneous firm capabilities to transform R\&D efforts into cost reductions. In particular, each firm exerts costly R\&D effort\footnote{Subject to decreasing marginal returns, in line with the evidence in \cite{bloom2020ideas}.}---with heterogeneous returns---to lower the marginal cost of production. Firms also form bilateral R\&D collaborations\footnote{We assume that pairwise link formation is costless, as in \cite{goyal2001r} and \cite{zirulia2012role}. Other work that considers linking costs includes \cite{goyal2003networks, galeotti2006network, westbrock2010natural, hsieh2025endogenous}.}  with each other to access the benefits of each other's R\&D efforts, while still competing à la Cournot in a homogeneous product market. This generalization to firm-level heterogeneity implies that, when firms form bilateral R\&D collaborations, the marginal benefit a firm derives from its partner’s effort depends on the partner’s productivity. That is, when a firm with relatively higher productivity increases its effort, the resulting cost reduction for its collaborators is substantial. In contrast, a lower-productivity firm contributes less to its partner’s cost reduction, even if it exerts the same level of R\&D effort. Such heterogeneity introduces asymmetries in equilibrium R\&D efforts and the firm-level gains from R\&D collaborations, thus unlocking new results on pairwise stable R\&D networks and new transmission channels of these efforts into firm profits and social welfare.

Under heterogeneity, there are two distinct effects of linking with another firm: (i) the productivity effect, and (ii) the connectivity effect---firm-level outcomes depend on both productivity differences and network externalities arising from R\&D alliances. To isolate the former effect, we first derive comparative statics on equilibrium efforts and profits for two firms with distinct productivities that occupy symmetric network positions, meaning they face the same network externalities. Our first result shows that when such two firms are connected, the low-productivity firm experiences a relatively larger increase in profit than the high-productivity firm while exerting a lower R\&D effort and therefore bearing a lower R\&D cost. Consequently, among two firms with symmetric network positions, whenever the high-productivity firm benefits from forming the link, the low-productivity firm also benefits, showing an ordering of asymmetric gains from linking in favor of the low-productivity firm. Furthermore, our simulation-based findings show that such asymmetric returns go beyond level effects, and in fact, create differences in link formation incentives: a low-productivity firm always benefits from connecting to a high-productivity firm, while the reverse does not necessarily hold. This implies there are positive/negative benefits from link formation for low-/high-productivity firms. 

Building on these fundamental insights, we study equilibrium network configurations under the widely adopted concept of pairwise stability introduced by \cite{jackson1996strategic} and employed in closely related R\&D network models, e.g., \cite{goyal2001r, goyal2003networks, zirulia2012role, billand2019firm}. Our next theoretical result shows that the complete network---which is pairwise stable in the homogeneous framework---ceases to be stable once the productivity gap between firms becomes sufficiently large. The reason is that the benefit of linking with another firm is monotonically increasing in that other firm’s R\&D productivity, which means there exists a productivity threshold for a given firm such that it does not form bilateral links with any firm that has a productivity level below that threshold.

Extending stability results to network structures beyond the complete network is known to be intractable due to network externalities, even under the homogeneous benchmark as discussed in \cite{goyal2001r}. This intractability is further compounded under heterogeneity. We therefore adopt a simulation-based approach to extract insights on the set of stable R\&D networks and their welfare implications.

In the rest of the analysis, for simplicity of exposition, we focus on the simplest form of heterogeneity by restricting our setting to two firm types: high- and low-productivity firms, with productivity levels normalized to $1$ and $\theta \in (0, 1]$, respectively. By using simulations, we then characterize the full set of pairwise stable networks for a given number of firms and type distribution, showing that, under heterogeneity, the set of stable configurations extends beyond the complete network. In particular, while the complete network continues to be stable when the productivity gap between high- and low-type firms is sufficiently small, if this gap is sufficiently large (low $\theta$), the class of Positive Assortative (PA) networks becomes stable. These networks are \textit{clustered} by productivity type, such that firms form bilateral links with all firms of the same type only. Specifically, we establish the existence of a lower threshold productivity $\theta$, denoted by $\underline{\theta}$, above which the complete network is stable, and an upper threshold, $\overline{\theta}$, below which the PA network is stable.  This implies that in cases where $\underline{\theta} < \overline{\theta}$, which holds when the fraction of high-productivity firms is not extremely large, an intermediate range for $\theta$ exists for which both networks are pairwise stable. Furthermore, the PA/complete network is the unique stable structure for sufficiently large/small productivity gaps.

The remainder of the paper shifts attention beyond firm-level outcomes to a broader implication of heterogeneity: aggregate welfare as measured by the sum of consumer and producer surplus. Simulation-based analysis uncovers the R\&D productivity puzzle: a higher average R\&D productivity in the economy ---achieved through a higher productivity of low-type firms ($\theta$) or a higher fraction of high-type firms ($\rho$)--- may exhibit lower welfare. This happens due to (i) the formation of alternative stable R\&D
network structures and/or (ii) a crowding-out effect of high-productivity firms. First, we find that the clustered structure (PA) delivers higher welfare compared to the complete network. Consequently, comparing two economies that differ in their productivity gaps, the economy with a smaller gap (in which the complete network is stable) exhibits lower welfare than the economy with a larger gap (in which the PA network is stable). Second, we find that welfare under PA formation follows an inverted U-shape in the fraction of high-productivity firms: welfare increases with this fraction at low levels, peaks at an intermediate fraction, then declines at high fractions, highlighting a crowding-out effect. 

These counterintuitive findings challenge the conventional economic wisdom that higher productivity is welfare-enhancing  \citep{romer1990endogenous, jones1995r, jones2000too, basu2022productivity}. Viewed through the lens of R\&D networks, we show that endogenous connectivity incentives can alter the welfare effects of R\&D productivity improvements. At a deeper level, our findings reveal that productivity affects both the efficiency and organization of innovation across firms. An important policy implication is that innovation-enhancing policies should account for their impact on R\&D alliances, as changes in stable network structures may offset or even reverse intended welfare gains. A natural avenue for future research is to test these mechanisms empirically, either in a cross-country setting or by examining changes in firms’ innovation productivity and collaboration patterns over time within a region. 

Our approach differs from existing contributions by allowing endogenous asymmetric benefits from link formation to arise from R\&D productivity gaps between firms. \cite{zirulia2012role} introduces exogenous, partner-specific spillovers, showing how firms with unique technological capabilities can become central hubs in their markets. Similarly, \cite{billand2019firm} studies connection-specific benefits, where the gains from linking are symmetric within a given pair but differ across types of pairings. While these models generate rich network structures, the marginal cost reduction from forming a link is symmetric for the two connected firms, regardless of underlying heterogeneity.

In contrast, our framework allows marginal cost reductions from collaboration to be endogenously asymmetric, generating fundamentally different incentives for link formation and leading to pairwise stable network structures with distinct welfare implications. Finally, whereas some related work studies the dynamic formation of R\&D networks and the evolution of technology diffusion \citep[e.g.,][]{bischi2012dynamicP1, hsieh2025endogenous}, our model focuses on a static environment, which allows us to analyze how heterogeneity shapes link formation, R\&D effort allocation, and aggregate welfare.

The remainder of the paper is as follows. Section \ref{sec:model} outlines the theoretical framework and the multi-stage game setup, providing our main theoretical findings. By simplifying to a two-type productivity setup, Section \ref{sec:simulation_results} provides the full set of pairwise stable networks for $n = 4$ firms, and extends insights on firm-level outcomes, pairwise stability and welfare through simulations. Section \ref{sec:conclusion} provides our closing remarks and discusses how our model can have implications for industrial policy.

\section{Model}
\label{sec:model}

There exists a set of firms  $\mathcal{N} \coloneqq \{1, \ldots, n\}$ competing in a single-product oligopolistic market and we employ a complete information multi-stage game setup similar to \cite{goyal2001r}. The timing of actions is as follows:

\begin{enumerate}

\item \textbf{Network Formation.} Firms engage in \textit{bilateral} R\&D alliances with one another, represented as an undirected link between two collaborators. Following \cite{jackson1996strategic}, we use the standard notion of pairwise stability as the equilibrium concept. Letting $\mathcal{G}^n$ denote the set of all possible binary $n \times n$ adjacency matrices, an R\&D network is represented by a zero-diagonal, symmetric adjacency matrix $\mathbf{G} \in \mathcal{G}^n$ in which $G_{ij}$ equals 1 if there is a collaboration between firms $i$ and $j$, and $0$ otherwise.

\item \textbf{R\&D Investment.} Given the network $\mathbf{G}$, any firm $i$ chooses their level of R\&D effort ($\Tilde{e}_i$) facing a quadratic cost $\text{C}(\Tilde{e}_i) = \phi_i \cdot \Tilde{e}_i^2$, where $\phi_i>0$ is the firm-specific cost function steepness. Firms then share their R\&D efforts with their collaborators in the network. As a result, each firm’s marginal cost of production ($c_i$) is:
\begin{equation}
    \label{eq:marginal_production_cost}
    c_i\left(\Tilde{\mathbf{e}}, \mathbf{G} \right) = \Bar{c} - \Tilde{e}_i - \sum_{j \in \mathcal{N}_i(\mathbf{G})} \Tilde{e}_j,
\end{equation}
where $\Bar{c}$ is the baseline (pre-R\&D) marginal cost common to all firms, and $\Tilde{\mathbf{e}} \coloneqq \left\{ \Tilde{e}_j \right\}_{j=1}^{n}$. We define $\mathcal{N}_i(\mathbf{G}) \coloneqq \{ j \mid G_{ij} = 1 \}$ as the set of firm $i$'s collaborators given the network $\mathbf{G}$.

\item \textbf{Product Market Competition.} Finally, firms compete in an oligopoly à la Cournot, where they produce a single homogeneous product, while facing different production costs under a linear demand curve characterized by
$
    p\left(\mathbf{q}\right) = \alpha - \sum_{i=1}^n q_i,
$
where $p$ is the price of the single product in the market, $\alpha$ represents the maximum price consumers are willing to pay, $q_i\in\mathbb{R}_+$ is firm $i$'s level of output, and we collect $\mathbf{q} \coloneqq \left\{ q_j \right\}_{j=1}^{n} $.
\end{enumerate}

In this three-stage game, firms maximize their profits $\pi_i$ by sequentially choosing collaboration links, then R\&D effort ($\Tilde{e}_i$), and then quantity ($q_i$), where:
\begin{equation}
    \label{eq:profits}
    \pi_i\left(\mathbf{q}, \Tilde{\mathbf{e}}, \mathbf{G} \right) = q_i \cdot \left[p\left(\mathbf{q} \right) - c_i\left(\Tilde{\mathbf{e}}, \mathbf{G} \right) \right] - \phi_i \cdot \Tilde{e}_i^2.
\end{equation}
We depart from the literature by adding heterogeneity in R\&D productivities into the baseline model described above as introduced by \cite{goyal2001r}. Specifically, while earlier studies assume that $\phi_i = \phi$ for all firms, in our framework $\phi_i$ captures the inverse of firm $i$’s R\&D productivity. A higher value of $\phi_i$ implies the firm must incur a greater R\&D cost to achieve a given efficiency (or technology) level as measured by its marginal production cost $c_i$. 

For ease of interpretation, we normalize R\&D costs relative to the most R\&D-efficient firm in the market. Let $\phi$ denote the cost coefficient of this most productive firm, defined as
$
	\phi \coloneqq \min\{\phi_1, \ldots, \phi_n\}.
$
We then define the \textit{relative R\&D productivity} of firm $i$ with respect to the most productive firm as $\theta_i \coloneqq \sqrt{\phi / \phi_i}$ and a rescaled \textit{R\&D effort} variable as $e_i \coloneqq \Tilde{e}_i / \theta_i$. Given this normalization, $\theta_i \in(0, 1]$ and the most productive firm(s) have $\max_{i \in \mathcal{N}} \theta_i = 1$.

These definitions enable us to rewrite the marginal cost function \eqref{eq:marginal_production_cost} and firm profits \eqref{eq:profits}:
\begin{align}
    \begin{aligned}
        \label{eq:marginal_cost}
        c_i\left(\mathbf{e}, \mathbf{G} \right) & = \bar{c} - \theta_i  e_i - \sum_{j\in \mathcal{N}_i(\mathbf{G})} \theta_j  e_j,  \\
        \pi_i\left(\mathbf{q}, \mathbf{e}, \mathbf{G} \right) & = q_i \cdot \left[p\left(\mathbf{q} \right) - c_i\left(\mathbf{e}, \mathbf{G} \right) \right] - \phi \cdot e_i^2,
    \end{aligned}
\end{align}
where $\mathbf{e} \coloneqq \{e_j\}_{j=1}^n$ is the vector of R\&D efforts. Henceforth, we refer to $e_i$ as firm $i$’s R\&D effort and to $\theta_i$ as its R\&D productivity. A higher $\theta_i$ corresponds to greater cost reductions, equivalent to a lower $\phi_i$ coefficient.

\subsection{Equilibrium Characterization}

We solve for the equilibrium production level, R\&D efforts, and collaboration links by performing backwards induction. First, we solve for the firm's maximization problem and obtain equilibrium production levels and R\&D efforts, corresponding to the last two stages of the game. To this end, we provide Proposition~\ref{prop:equilibrium_existance} that guarantees the existence of a solution with strictly positive efforts and outputs. 

These last two stages can be interpreted as the solution to a game with an exogenously given R\&D network \citep[a similar structure is analyzed in][under the assumptions of homogeneous productivity and multiple markets]{konig2019r}. Finally, we solve for the equilibrium network structure invoking the concept of pairwise stability \citep{jackson1996strategic}, characterizing the endogenous formation of R\&D alliances for heterogeneous firms.

\subsubsection{Market Competition}
In the final (third) stage of the game, firms choose their production quantities conditional on their realized marginal costs $c_i$, taking the market demand function as given. The equilibrium production quantities and corresponding profits are given by\footnote{The derivations are provided in the \href{https://drive.google.com/file/d/1nduNFpbgYaX7XRpILcPp3MSKFv4mXkue/view?usp=sharing}{Supplementary Appendix}.}
\begin{align}
    \label{eq:production_amount_eq}
    \begin{aligned}
        \str{q}_i\left(\mathbf{e}, \mathbf{G} \right) & = \frac{1}{n+1} \left[ \alpha - n \cdot c_i\left(\mathbf{e}, \mathbf{G} \right) + \sum_{j \neq i} c_j\left(\mathbf{e}, \mathbf{G} \right) \right], \\
        \str{\pi}_i\left(\mathbf{e}, \mathbf{G} \right) & = \str{q}_i\left(\mathbf{e}, \mathbf{G} \right)^2 - \phi \cdot e_i^2. 
    \end{aligned}
\end{align}
Let $d_i(\mathbf{G}) \coloneqq |\mathcal{N}_i(\mathbf{G})| = \sum_{j\in\mathcal{N}} G_{ij}$ denote the degree of firm $i$ in network $\mathbf{G}$. Substituting equation~\eqref{eq:marginal_cost} into \eqref{eq:production_amount_eq}, the equilibrium output level of firm $i$ can be rewritten as a function of the R\&D collaboration network and effort profile:
\small
\begin{equation}
\label{eq:quantity_with_desc_n_nn}
    \str{q}_i \left(\mathbf{e}, \mathbf{G}\right)
	= 
	\underbrace{\frac{\alpha-\bar{c}}{n+1}}_{\text{baseline}}
	+ \underbrace{\frac{n - d_i(\mathbf{G})}{n+1}\theta_i e_i}_{\text{own effort effect}}
	+ \underbrace{\sum_{j\in \mathcal{N}_i(\mathbf{G})} \frac{n - d_j(\mathbf{G})}{n+1}\theta_j e_j}_{\text{neighbors' effort effect}}
	- \underbrace{\sum_{k\in \mathcal{N}_{-i}(\mathbf{G})}\frac{1 + d_k(\mathbf{G})}{n+1}\theta_k e_k}_{\text{non-neighbors' effort effect}},
\end{equation}
\normalsize
where $\mathcal{N}_{-i}(\mathbf{G}) \coloneqq \left\{j\in\mathcal{N} \mid G_{ij} = 0, j \neq i \right\}$ denotes the set of firms that are not neighbors of firm $i$ in the network $\mathbf{G}$. Equation~\eqref{eq:quantity_with_desc_n_nn} shows that a firm’s optimal production level increases with its own R\&D effort and decreases with the R\&D efforts of its non-neighbors. The effect of R\&D efforts by firm $i$'s neighbors on its production is positive given that the degree of each firm $j \in \mathcal{N}$ always satisfies $0 \leq d_j(\mathbf{G}) \leq n-1 < n$, while the effect from non-neighbours is therefore always negative.

\subsubsection{Equilibrium R\&D Efforts}

In the second stage, we solve for each firm's optimal R\&D effort $e_i$, conditional on the collaboration network $\mathbf{G}$. Using \eqref{eq:production_amount_eq}, the first order condition (FOC) for profit maximization with respect to $e_i$ is:
\begin{equation*}
	\left. \frac{ \partial \str{\pi}_i\left(\mathbf{e}, \mathbf{G}\right) }{ \partial e_i } \right|_{\str{e}_i} =  \quad \left(\left. \frac{2 \str{q}_i(\mathbf{e}, \mathbf{G})}{n+1} \right|_{\str{e}_i} \right) \cdot \left( -n \left.\frac{ \partial c_i(\mathbf{e}, \mathbf{G}) }{ \partial e_i }\right|_{\str{e_i}} + \sum_{j \ne i} \left. \frac{ \partial c_j(\mathbf{e}, \mathbf{G}) }{ \partial e_i } \right|_{\str{e_i}} \right ) - 2 \phi  \str{e}_i = 0,
\end{equation*}
where $\str{e}_i\left(\mathbf{e}_{-i}, \mathbf{G}\right) \coloneqq \arg\max_{e_i}\{\str{\pi}_i\}$ is firm $i$’s best-response effort level, and $\mathbf{e}_{-i} \coloneqq \{e_j\}_{j\ne i}$ is the effort of all other firms. Let $\eta_i\left(\mathbf{G}\right)\coloneqq [n- d_i\left(\mathbf{G}\right)] / (n+1)$ represent a strictly positive \textit{sparsity} coefficient that decreases with the firm's degree $d_i(\mathbf{G})$. From equation \eqref{eq:marginal_cost}, we know that 
$\partial c_i(\mathbf{e}, \mathbf{G})/\partial e_i = -\theta_i$ and $\partial c_j(\mathbf{e}, \mathbf{G})/\partial e_i = -\theta_i G_{ij}$, such that we can solve the FOC in order to obtain:
\small
\begin{equation}
	\label{eq:foc_effort_desc}
	\str{e}_i \left(\mathbf{G}, \mathbf{e}_{-i}\right) = 
	\underbrace{\frac{\theta_i\eta_i(\mathbf{G})}{\phi-\theta_i^2\eta_i^2(\mathbf{G})}}_{\text{\scriptsize (i) sparsity}}
	\Biggl\{
	\underbrace{\frac{\alpha - \bar{c}}{n+1}}_{\text{\scriptsize (ii) baseline}}
	+  \underbrace{\sum_{j\in\mathcal{N}_i(\mathbf{G})} \eta_j(\mathbf{G})  \theta_j e_j}_{\text{\scriptsize (iii) neighbor effect}}
	-  \underbrace{\sum_{k\in \mathcal{N}_{-i}(\mathbf{G})}\big[1 - \eta_k(\mathbf{G})\big] \theta_k  e_k}_{\text{\scriptsize (iv) non-neighbor effect}}
	\Biggr\},
\end{equation}
\normalsize
The best-response effort of firms can be decomposed into four key factors. The first is a proportionality constant that is increasing in the firm's productivity $\theta_i$ and its sparsity coefficient $\eta_i(\mathbf{G})$, and therefore decreasing in its degree $d_i(\mathbf{G})$. The second is a baseline effort level that depends on market parameters $\alpha$, $\Bar{c}$, and $n$. The structures of the neighbor effect (iii) and non-neighbor effects (iv) mirror that in equation~\eqref{eq:quantity_with_desc_n_nn}, indicating that a firm's optimal R\&D effort is increasing in neighbors' efforts and productivities, and decreasing in non-neighbors' efforts and productivities as well as the degree of all competitor firms.

Letting $\mathbf{1}_n$ denote an $n \times 1$ vector of ones, the FOCs in \eqref{eq:foc_effort_desc} yields a system of $n$ linear equations in $n$ unknowns:
\begin{equation}
	\label{eq:optimal_effort_matrix}
	\mathbf{A}(\mathbf{G})  \mathbf{e} = (\alpha - \bar{c}) \mathbf{1}_n, \quad \text{where }
	A_{ij}(\mathbf{G})  =  \begin{cases}
		\frac{(n+1)^2\phi}{\theta_i[n - d_i(\mathbf{G})]} - \theta_i[n - d_i(\mathbf{G})] & \text{for } i = j,\\
		\left[1 + d_j(\mathbf{G})\right]\theta_j - (n+1)G_{ij}\theta_j & \text{for } j \neq i.
	\end{cases}
\end{equation}
\begin{proposition}[Cost parameter bound]
	\label{prop:equilibrium_existance}
	There exists a unique optimal effort profile $\str{\mathbf{e}}(\mathbf{G}) \coloneqq \{\str{e}_j(\mathbf{G})\}_{j=1}^{n}$ with $\str{e}_i(\mathbf{G}) > 0$ for all $i \in \mathcal{N}$ that solves equation~\eqref{eq:optimal_effort_matrix} if the cost parameter $\phi$ satisfies:
	\begin{equation}
		\label{eq:phi_condition}
		\phi > \frac{n[2(n-1)^2 + n]}{(n+1)^2} \eqqcolon \underline{\phi}.
	\end{equation}
\end{proposition}
\vspace{-6mm}
Proposition~\ref{prop:equilibrium_existance} establishes that a unique and interior equilibrium exists for any network $\mathbf{G} \in \mathcal{G}^n$, provided the cost coefficient $\phi$ exceeds a threshold $\underline{\phi} > 2(n - 1)$.
The proof of Proposition~\ref{prop:equilibrium_existance} relies on both the Levy-Desplanques theorem \citep[itself a consequence of Gershgorin's disc theorem,][]{gershgorin1931uber} and matrix power series \citep[e.g., Chapter 5 of][]{Horn_Johnson_2012}. Specifically, we provide a bound for $\phi$ that guarantees the matrix $\mathbf{A}(\mathbf{G})$ is strictly diagonally dominant (which guarantees its non-singularity) and its inverse has positive row sums regardless of the network $\mathbf{G}$. Both facts together imply the FOC system \eqref{eq:optimal_effort_matrix} admits a unique solution with strictly positive R\&D efforts.

While existence and positivity of a unique equilibrium effort profile are guaranteed by Proposition~\ref{prop:equilibrium_existance}, equilibrium effort \emph{levels} cannot be generally expressed in closed form, as they involve an intractable inverse of the non-symmetric matrix $\mathbf{A}(\mathbf{G})$. Nevertheless, we can derive equilibrium profits using equilibrium efforts as:
\begin{equation}
    \label{eq:eq_profit_effort}
    \str{\pi}_i \left(\mathbf{G}\right) = \left[\frac{\phi}{\theta_i^2\eta_i^2(\mathbf{G})} - 1\right]\phi \cdot {\str{e}_i}^2\left(\mathbf{G}\right).
\end{equation}
To uncover the black box of equilibrium efforts and profits under heterogeneous productivities, we characterize relative equilibrium efforts and profits of any two firms $i$ and $j$ that are symmetric with respect to their position in the network. This is formalized in the following definition:
\begin{definition}[Symmetric position]
    Firms $i, j\in\mathcal{N}$ have a \emph{symmetric position} in a given network $\mathbf{G}$ if $\forall k\in\mathcal{N}\ne i, j$, $G_{ik} = G_{jk}$. 
    \label{definition_symmetry}
\end{definition}
\vspace{-4mm}
For two firms $i$ and $j$ that have a symmetric position, they hold identical links in the network $\mathbf{G}$. This directly implies their degrees are also the same, $d_i(\mathbf{G}) = d_j(\mathbf{G})$, and so are their sparsity coefficients $ \eta_i(\mathbf{G}) = [n - d_i(\mathbf{G})] / (n+1) =[n - d_j(\mathbf{G})] / (n+1)  = \eta_j(\mathbf{G})$.

\begin{proposition}[Effort and profit ratios under symmetric positions]
    \label{prop:summetric_ratio}
    Let firms $i, j\in\mathcal{N}$ have a symmetric position in a given network $\mathbf{G}$. Then, the ratio of their equilibrium efforts and profits are:
    \begin{equation}
        \label{eq:ratio_effort_profit_prop2}
        \begin{aligned}
        \frac{\str{e}_i(\mathbf{G})}{\str{e}_j(\mathbf{G})} & = \frac{\theta_i}{\theta_j}  \frac{\phi - \theta_j^2\eta_j(\mathbf{G})(1-G_{ij})}{\phi - \theta_i^2\eta_i(\mathbf{G})(1-G_{ji})},
        \\
        \frac{\str{\pi}_i(\mathbf{G})}{\str{\pi}_j(\mathbf{G})} & = \frac{\phi - \theta_i^2 \eta_i^2(\mathbf{G})}{\phi - \theta_j^2\eta_j^2(\mathbf{G})} \left[ 
        \frac{\phi - \theta_j^2\eta_j(\mathbf{G})(1-G_{ij})}{\phi - \theta_i^2\eta_i(\mathbf{G})(1-G_{ji})}
        \right]^2.
         \end{aligned}
    \end{equation}
\end{proposition}
\vspace{-6mm}
Proposition \ref{prop:summetric_ratio} provides an expression for the ratio of optimal efforts and profits for symmetrically positioned firms $i$ and $j$. It shows that if $\theta_i>\theta_j$, then $\str{e}_i(\mathbf{G})>\str{e}_j(\mathbf{G})$ always holds, regardless of the connection between $i$ and $j$. That is, if firm $i$ is more productive than firm $j$ while having the same network position, then firm $i$ commits higher R\&D effort than $j$.
In particular, let $\mathbf{G}_{+ij}$ denote the original network $\mathbf{G}$ where the link between firms $i$ and $j$ is present (i.e., $G_{ij} = G_{ji} = 1$), with all other links remaining unchanged; and let $\mathbf{G}_{-ij}$ denote the network where this link is not present (i.e., $G_{ij} = G_{ji} = 0$). Then, the relative effort ratio given in Proposition~\ref{prop:summetric_ratio} implies:
\begin{equation*}
    1 < \frac{\theta_i}{\theta_j} = \frac{\str{e}_i(\mathbf{G}_{+ij})}{\str{e}_j(\mathbf{G}_{+ij})} < \frac{\str{e}_i(\mathbf{G}_{-ij})}{\str{e}_j(\mathbf{G}_{-ij})},
\end{equation*}
meaning that when the link between $i$ and $j$ is removed, the relative effort of the higher productive firm $i$ compared to the lower productive firm $j$ rises.

The ranking of equilibrium profits based on their productivities, however, does depend on whether these two firms form a collaboration link or not. 
To better understand the incentives for forming a link between two firms with different productivity levels\textemdash under the assumption that they occupy symmetric positions in the network\textemdash we next compare their relative profits with and without such a link.
\begin{corollary}
    \label{corrolary_profit}
    Let $i, j\in \mathcal{N}$ have a symmetric position in $\mathbf{G}$, and $\theta_i > \theta_j$.
    Then:
    \begin{enumerate}
        \item[(i)] $\str{\pi}_i(\mathbf{G}_{-ij}) > \str{\pi}_j(\mathbf{G}_{-ij})$, that is, when $i$ and $j$ are not connected, the firm with higher productivity has higher profit.
        \item[(ii)] $\str{\pi}_i(\mathbf{G}_{+ij}) < \str{\pi}_j(\mathbf{G}_{+ij})$, that is, when $i$ and $j$ are connected, the firm with higher productivity has lower profit.
    \end{enumerate}
\end{corollary}
\vspace{-4mm}
Defining the link deviation operator $\Delta^{ij}$ on any function $f: \mathcal{G}^n\to\mathbb{R}$ as $
\Delta^{ij}f(\mathbf{G}) \coloneqq f (\mathbf{G}_{+ij}) - f(\mathbf{G}_{-ij})$, Corollary \ref{corrolary_profit} also implies that if $\Delta^{ij}\str{\pi}_i(\mathbf{G}) >0$ holds, then $\Delta^{ij}\str{\pi}_j(\mathbf{G}) > 0$ also holds. This means that if adding the link between $i$ and $j$ is beneficial (profit-increasing) for the firm with higher productivity, then it is always beneficial for the firm with lower productivity.
However, whether each such firm benefits from forming a link is still ambiguous, and this problem remains intractable even under stronger assumptions on the network structure (e.g., complete or empty network) than the symmetric position assumption we use. Therefore, in Section \ref{sec:simulation_results} we take a simulation-based approach to uncover relevant theoretical insights into pairwise stability.

Finally, we provide a comparison of firm sizes as proxied by output levels. Rearranging the FOCs, we get the optimal quantity produced by any firm $i$ as:
\begin{equation}
	\label{eq:optimal_quant_effort}
    \str{q}_i\left(\mathbf{G}\right) = \frac{\phi}{\theta_i\eta_i(\mathbf{G})}
    \str{e}_i \left(\mathbf{G}\right).
\end{equation}
This can be compared with the result in \citet{hsieh2025endogenous}, showing $\str{e}_i = \theta \str{q}_i/(2\phi) $ (with the relevant changes in notation). The differences in equilibrium effort and profit levels in these two setups arise from, first, differences in sources of heterogeneity (in baseline costs in their model vs. R\&D productivities in ours) and, second, the solution methods implemented for equilibrium construction (simultaneous in their model vs. sequential action in ours). Equation~\eqref{eq:optimal_quant_effort} implies that a firm's equilibrium size is positively related to its equilibrium R\&D effort, consistent with empirical findings in \cite{cohen1996firm}. Additionally, we see from the expression that the optimal output level is negatively related to its productivity $\theta_i$ and its sparsity coefficient $\eta_i(\mathbf{G})$. 

Equations \eqref{eq:ratio_effort_profit_prop2}
and	\eqref{eq:optimal_quant_effort} imply that for comparing one high- to one low-productive firm with $\theta_i>\theta_j$:
\begin{equation}
    \label{eq:size_ratio}
    \frac{\str{q}_i(\mathbf{G}_{+ij})}{\str{q}_j(\mathbf{G}_{+ij})} = 1,
    \quad \text{ and } \quad
    \frac{\str{q}_i(\mathbf{G}_{-ij})}{\str{q}_j(\mathbf{G}_{-ij})} = \frac{\phi - \theta_j^2 \eta_j(\mathbf{G}_{-ij})}{\phi - \theta_i^2 \eta_i(\mathbf{G}_{-ij})} > 1.
\end{equation}
Equation~\eqref{eq:size_ratio} shows that although the less productive firm earns a higher profit than the more productive firm when they are connected, this ranking does not extend to output levels. Under the symmetric position assumption, the more productive firm always produces a weakly higher output, and strictly higher when the two firms are not connected. This result clarifies why the more productive firm may earn lower profits under connection: both firms face the same price and marginal cost of production leading to the same production level and therefore equal profits in the competitive market. However, the more productive firm incurs a higher R\&D cost due to its greater effort level, as shown in Equation~\eqref{eq:ratio_effort_profit_prop2}, leading to an overall lower profit.

\subsubsection{R\&D Network Formation}

Pairwise stability requires that no firm has an incentive to sever an existing link, and that no two firms both benefit (with at least one strictly gaining) from forming a new one \citep{jackson1996strategic}. In our case, the formal statement translates to the following: a network $\mathbf{G}$ is \textit{pairwise stable} if and only if for all pairs of firms $(i, j)$ where $i, j\neq i \in \mathcal{N}$:
\begin{itemize}
    \item \textit{Existing Alliances}: If ${G}_{ij} = 1$, then neither firm $i$ nor firm $j$ would benefit from dissolving the alliance, i.e., $\str{\pi}_i(\mathbf{G}) \ge \str{\pi}_i(\mathbf{G}_{-ij})$ and $\str{\pi}_j(\mathbf{G}) \ge \str{\pi}_j(\mathbf{G}_{-ij})$.
   
    \item \textit{Potential Alliances}: If ${G}_{ij} = 0$, then at least one of the firms would incur a loss in profit by forming the alliance, i.e., $\pi_i(\mathbf{G}_{+ij}) < \pi_i(\mathbf{G})$ or  $\pi_j(\mathbf{G}_{+ij}) < \pi_j(\mathbf{G})$.
\end{itemize}

Let the complete network $\mathbf{G}^C$ be the one in which all firms are linked to every other firm ($\forall i, j \in \mathcal{N}$ with $i \neq j$, $G^C_{ij} = 1$). Under homogeneity, \cite{goyal2001r} show that the complete network is pairwise stable. By generalizing to heterogeneous R\&D productivities, we present our next result.

\begin{proposition}[Thresholds for complete network stability]
    \label{prop:FC_stable}
    For any productivity distribution $\{\theta_i\}_{i \in \mathcal{N}} \in (0,1]^n$ with $\max_{i \in \mathcal{N}} \theta_i = 1$, there exist $\theta^{*},\theta^{**} \in (0, 1)$ with $\theta^{*} \le \theta^{**}$ such that the complete network $\mathbf{G}^C$ satisfies:
    \begin{enumerate}[label=(\roman*)]
        \item If $\exists j\in\mathcal{N}$ for which $ \theta_j < \theta^{*}$, then $\mathbf{G}^C$ is not pairwise stable.
        \item If $\theta_j \ge \theta^{**}$ for all $j \in \mathcal{N}$, then $\mathbf{G}^C$ is pairwise stable.
    \end{enumerate}
\end{proposition}

Proposition~\ref{prop:FC_stable} shows that when all firms are sufficiently similar in productivity, the complete network remains stable, as in the homogeneous benchmark. However, when productivity differences are large, highly productive firms no longer find it profitable to maintain links with much less productive partners. In this case, the complete network becomes unstable because the most productive firms prefer to sever some of their connections. The key intuition is simple. Linking to a less productive firm yields little benefit for a highly productive firm but still requires sharing the gains from its own R\&D effort. As the productivity gap widens, this imbalance grows, eventually making some links unattractive for the most productive firms. Stability therefore breaks down once the productivity of any firm falls below a threshold.

\subsection{Simulation preview: Link sustainability for random networks}

To generalize insights on pairwise stability beyond the complete network setting in Proposition \ref{prop:FC_stable}, we provide simulation-based evidence in Figure~\ref{fig:random_deviation} of link sustainability in random networks. In these simulations, we fix an arbitrary firm $i$ with productivity $\theta_i$ and consider a firm $j \neq i$ with a productivity $\theta_j$, identifying pair-specific threshold $\str{\theta}_{ij}$ values such that $i$ would not form a collaboration with $j$ if $\theta_j < \str{\theta}_{ij} < \theta_i$; i.e., if firm $j$ is sufficiently less productive than firm $i$. These pair-specific threshold values now depend on both the productivity distribution ($\beta$-distribution panels above) and the random network structure $\mathbf{G}$ with varying connection densities shown in different colors. 

In Figure \ref{fig:random_deviation}, we simulate random networks of $n = 20$ firms using an Erd\H{o}s-R\'{e}nyi generation \citep{erdHos1960evolution}, with probability of connection $\ell \in [0, 1]$ (where $\ell = 0$ corresponds to the empty network and $\ell = 1$ to the complete network). We obtain productivity distributions $(\theta_1, \ldots, \theta_n)$ as random samples from Beta distributions with varying parameters. Finally, we fix the productivities of two firms $i$ and $j$, having $\theta_i \in \{0.25, 0.5, 0.75\}$ (in three different rows) and allow $\theta_j$ to vary between $0$ and $\theta_i$ ($x$-axes values), and show how severing the link between such $i$ (solid line) and $j$ (dashed line) affects their profits, while changing the density of the network and the productivity distribution for other firms. Figure~\ref{fig:random_deviation} plots the productivity distributions in the upper panel that are used to calculate profit changes $\Delta^{ij} \str{\pi}(\mathbf{G}) / \str{\pi}(\mathbf{G}_{-ij})$ in the main (lower) panel, also varying the level of network connectivity $\ell$ within each case.

\begin{figure}[htbp]
    \centering
    \caption{Percentage changes in profits of two arbitrary firms after forming an R\&D collaboration link between them in random networks with varying connection density and firm productivity distributions}
    \includegraphics[width=\linewidth]{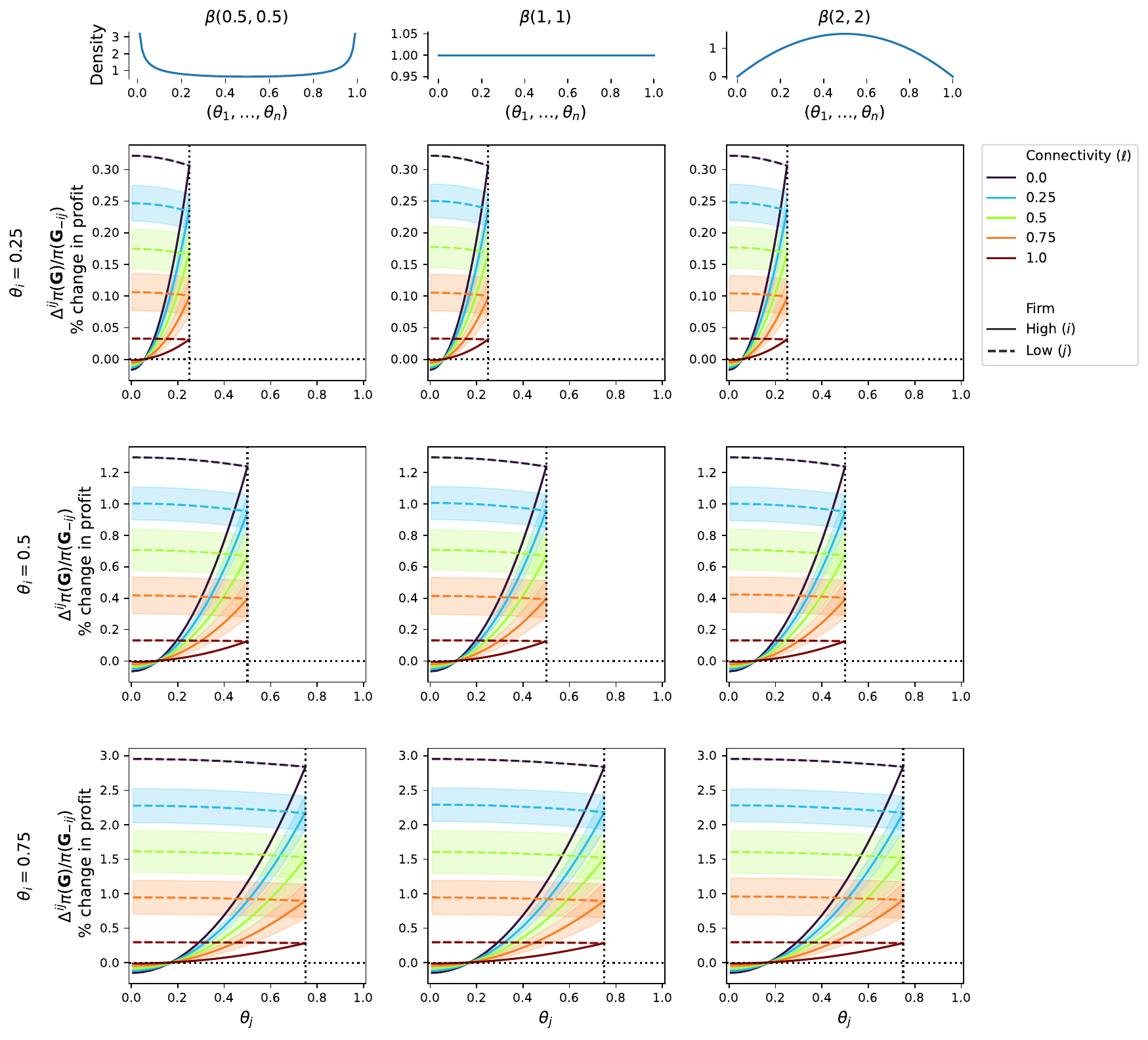}
    \label{fig:random_deviation}
    \fignote{Networks $\mathbf{G}$ with $n = 20$ firms generated as random binary adjacency matrices according to an Erdos-Renyi scheme with connectivity parameter $\ell \in \{0,0.25, 0.5, 0.75, 1\}$, represented by different colors. Productivity distributions $(\theta_1, \ldots, \theta_n)$ drawn as random samples from a Beta$(0.5, 0.5)$ distribution \textit{(left)}; Beta$(1, 1)$ \textit{(middle)}; and Beta$(2, 2)$ \textit{(right)}, with corresponding density plotted in the panels above. We set $\phi = \underline{\phi}$ for all networks. The percentage change in profit of firm $i$ (before and after adding a link to firm $j$) is drawn using a solid line, and for firm $j$ using a dashed line. Productivity of the higher-productive firm $i$ is chosen from $\theta_i \in \{0.25, 0.5, 0.75\}$. Percentage change in profits for firm $i$ crosses 0 at a threshold $\str{\theta}_{ij} \in (0, 1)$ for each productivity distribution and connectivity level cases shown. Figures are limited to the range $0 < \theta_j < \theta_i$ of each configuration to highlight the crossing at 0 for the higher-productive firm $i$.}
\end{figure}

\section{Simulation exercises for stability and welfare}
\label{sec:simulation_results}

The analysis so far shows that the complete network need not remain stable when productivities are heterogeneous, unlike \cite{goyal2001r}, where they also show that this structure 
is not welfare-maximizing in the sense that less connectivity is preferred in the social optimum. Building on these insights, in this section, we extend the analysis to investigate both the stability and welfare properties of alternative network structures in the heterogeneous-productivity setting.

Due to the intractable nature of the problem discussed earlier, we are unable to fully characterize the set of pairwise stable networks in general. Under homogeneity, \citet{goyal2001r} encounter similar difficulties in determining whether the complete network is the unique stable (symmetric) configuration, and show this is the case by restricting the problem to $n=4$ firms.\footnote{Readers are referred to pp. 696 of \citet{goyal2001r} for a detailed discussion on the intractability of the network formation problem in this class of models.} This challenge becomes much more pronounced under heterogeneity. Additionally, this intractability extends to making comparisons of equilibrium effort and profit levels across different network structures. Even for networks in which the matrix $\mathbf{A}(\mathbf{G})$ can be inverted to provide a closed-form solutions for efforts via Equation \eqref{eq:optimal_effort_matrix}, comparing effort levels across structures requires comparing rational functions of large degree, which is again unfeasible even for small number of firms $n$.

To open this black box, we categorize firms into two types based on their R\&D productivities: \textit{high}- and \textit{low}-productive firms— a simplified structure that makes the simulation exercise tractable for the remainder of the paper. The productivity of a high-type firm is therefore $\theta_H = 1$, and for a low-type firm it is denoted by $\theta \coloneqq \theta_L \in (0, 1)$. We also define $n_L$ and $n_H$ as the number of low- and high-type firms in the economy, respectively, such that $\rho \coloneqq n_H / n$ is the ratio of high-type firms. We set $\rho = 1/2$ for the initial simulations, and then we experiment with its value in extensions highlighting a novel crowding-out mechanism on welfare operating through high-type firms.

We begin by fully characterizing the set of stable networks for $n=4$ firms, by using computer-assisted derivations for the closed form solutions.\footnote{The full implementation of the computer-assisted derivations are in the \href{https://drive.google.com/file/d/1nduNFpbgYaX7XRpILcPp3MSKFv4mXkue/view?usp=sharing}{Supplementary Appendix}.} Crucially, the analysis reveals that differences in productivity levels and the distribution of high- and low-type firms within the economy can sustain pairwise stable networks different from the complete network.

Our first simulation documents a new fact: when introducing heterogeneity in firm R\&D productivities, the \textit{positive assortative} (PA) network $\mathbf{G}^{PA}$ ---in which firms fully connect only to the others of their same type--- becomes a pairwise stable configuration. In a two-types setting, a PA network is the one in which high- (low-) type firms are only linked to other high- (low-) type firms ($G_{ij}^{PA} = G_{ji}^{PA} = 1$ if $i$ and $j$ are of the same type, otherwise $G_{ij}^{PA} = G_{ji}^{PA} = 0$). This structure leads to two fully connected components of size $\rho \cdot n$ and $(1 - \rho) \cdot n$, composed solely of firms of high- and low-type firms, respectively.

\begin{figure}[htbp]
	\centering
	\caption{Pairwise stability domains for $n=4$ firms with two high- and two low-productivity firms for different productivity gap levels}
	\label{fig:stability_2H2L}
	\includegraphics[width=\linewidth]{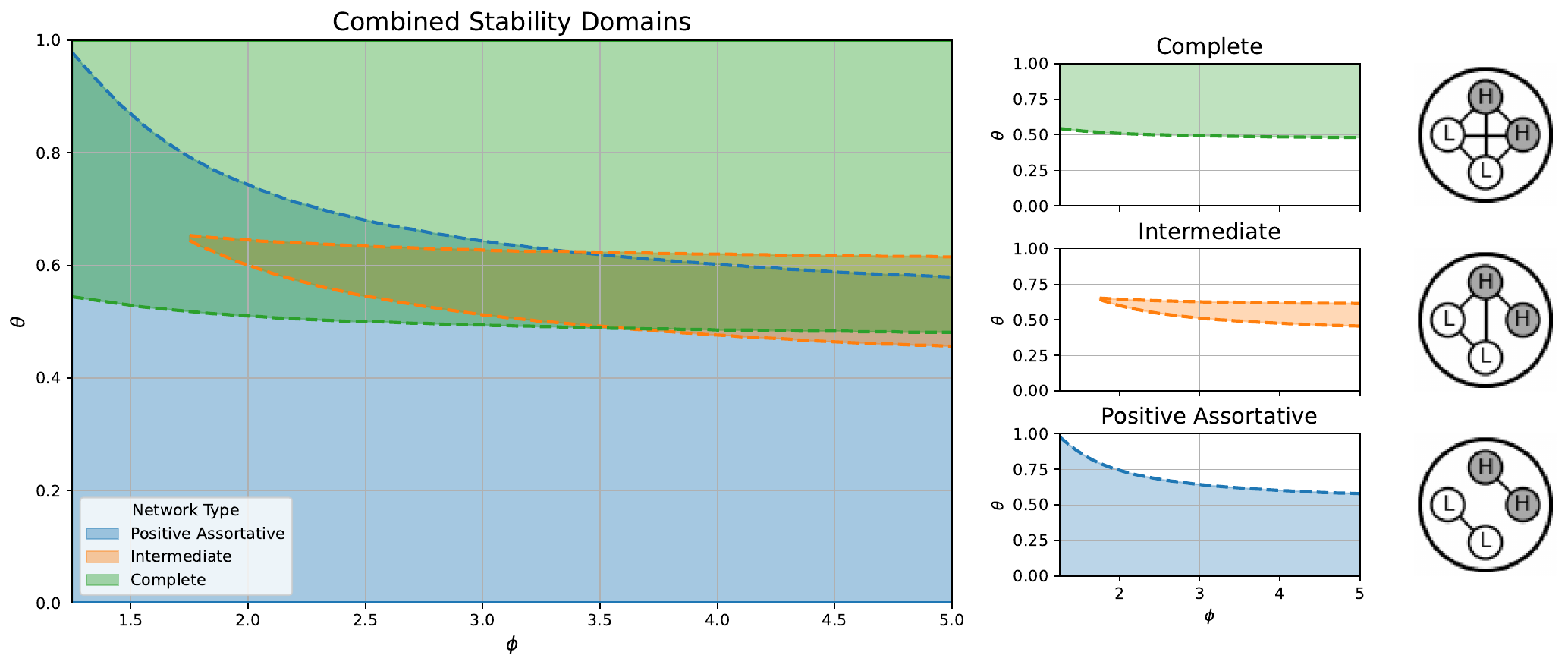}
	\fignote{The shaded region indicates the parameter combinations $(\theta, \phi)$ under which the corresponding configuration is pairwise stable. This figure only plots those network structures that showed non-empty stability region in simulations. The figure for every possible structure is provided in the \href{https://drive.google.com/file/d/1nduNFpbgYaX7XRpILcPp3MSKFv4mXkue/view?usp=sharing}{Supplementary Appendix}.}
\end{figure}

Figure \ref{fig:stability_2H2L} fully characterizes the set of pairwise stable networks for $n=4$ and $\rho=1/2$. To construct this figure, we solve for the closed-form equilibrium profits of each starting possible network configuration (in a space of $2^6$ starting structures) and numerically evaluate all possible link deviations from a given starting structure. We plot the contour of values for $\phi$ and $\theta$ that guarantee no profitable deviations exists from a starting network, which precisely corresponds to the full characterization of pairwise stable configurations.

Our results show that, for any fixed R\&D effort cost coefficient $\phi$, there exist threshold levels $\underline{\theta}$ (shown as the green dashed line) and $\overline{\theta}$ (shown as the blue dashed line) with $\underline{\theta} < \overline{\theta}$. The complete network is stable for $\underline{\theta} \le\theta <1$, and the PA network is stable for $0<\theta \leq \overline{\theta}$. Together, these facts imply both networks are stable in the interim region $\underline{\theta} \leq \theta  \leq \overline{\theta}$. In addition, both thresholds are decreasing in $\phi$, meaning that the stability area spanned by the complete/PA network is increasing/decreasing in R\&D cost.

Furthermore, while the PA and complete networks span the entire parameter space for $\theta$, there is one other configuration that emerges as stable for certain parameter values, which is shown in the middle right panel of Figure  \ref{fig:stability_2H2L}. This network is an intermediate between PA and complete, in the sense that it is denser/sparser than the PA/complete network through adding/severing links from a high-type firm to all low-type firms. Lastly, this characterization shows that for sufficiently low/high $\theta$, the PA/complete network is the unique stable structure.

\subsection{Welfare Comparison}

Following \cite{goyal2001r}, we define welfare as the total surplus generated in the economy by adding producer and consumer surplus.\footnote{For an alternative approach based on a utility function over consumption, see \cite{hsieh2025endogenous}.} Producer surplus ($PS$) is the sum of profits earned by all firms in the network, given by
$
    PS \coloneqq \sum_{i=1}^n \str{\pi}_i,
$
and consumer surplus ($CS$) is given by
$
    CS \coloneqq (1/2) \left(\sum_{i=1}^n \str{q}_i\right)^2,
$
as derived from the linear demand specification $p(\mathbf{q})$. Adding up these components, the total welfare of any given R\&D network structure $\mathbf{G}$ is given by:
\begin{equation}
    \label{eq:welfare_define}
    W(\mathbf{G}) = \frac{1}{2} \left[\sum_{i\in \mathcal{N}} \str{q}_i(\mathbf{G}) \right]^2 + \sum_{i\in\mathcal{N}} \str{\pi}_i(\mathbf{G}).
\end{equation}
To compare the welfare of the stable network configurations that arise in our setting, as well as R\&D efforts and profits of firms, in Figure \ref{fig:n6_profit_effort_stability} we next provide a simulation result increasing the number of firms to $n=6$ while keeping the two-type firm setting fixed.\footnote{When $n$ varies, we set $\phi = \underline{\phi}$, the lower bound of R\&D cost that generates positive efforts as described in Proposition \ref{prop:equilibrium_existance}.}

\begin{figure}
    \caption{Welfare, profit, and effort comparison of pairwise stable structures in $n=6$ setting, with $\rho=1/2$}
    \label{fig:n6_profit_effort_stability}
    \begin{subfigure}[t]{\linewidth}
        \centering
        \includegraphics[width=\linewidth]{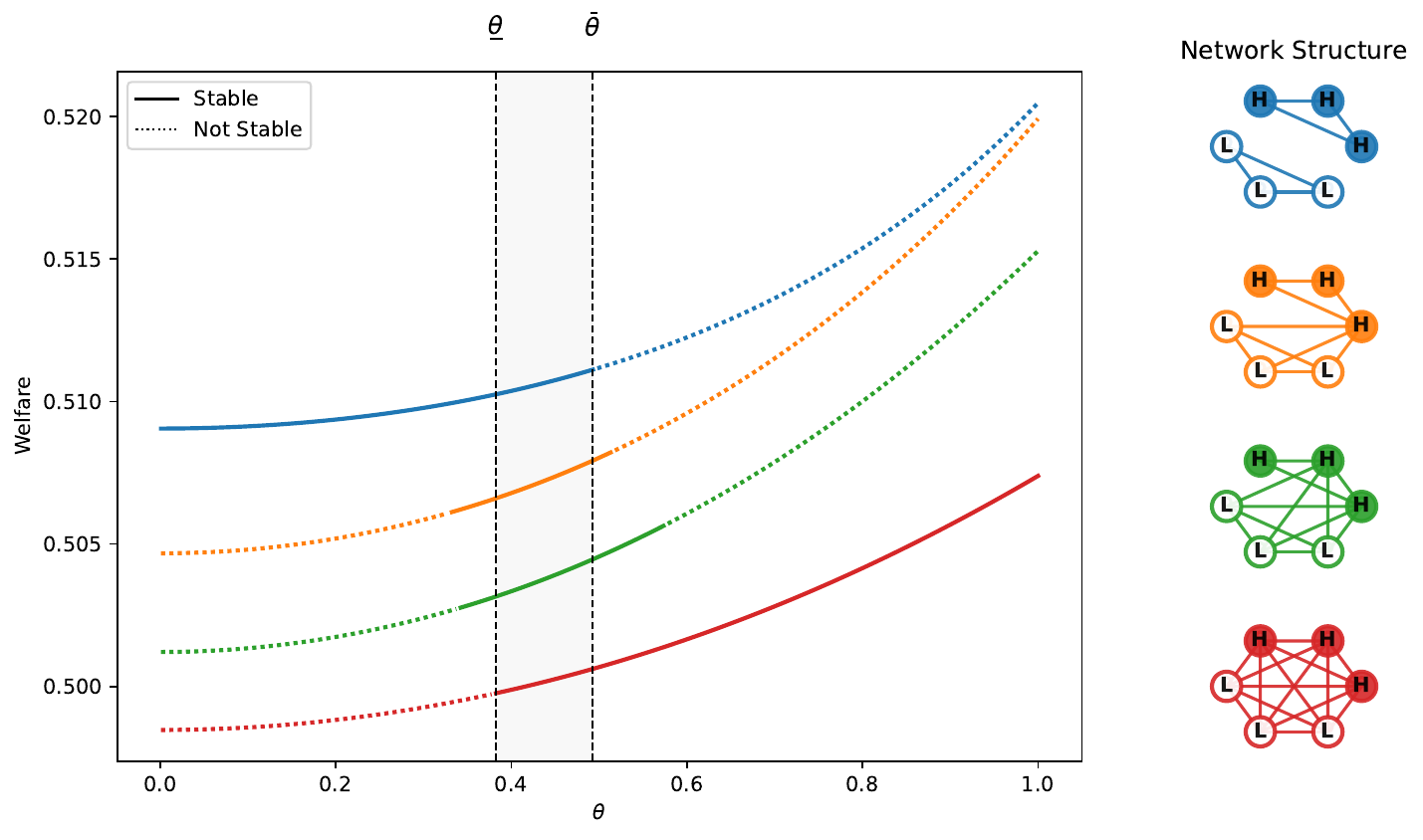}
    \end{subfigure}
    \begin{subfigure}[b]{\linewidth}
        \centering
        \includegraphics[width=1\linewidth]{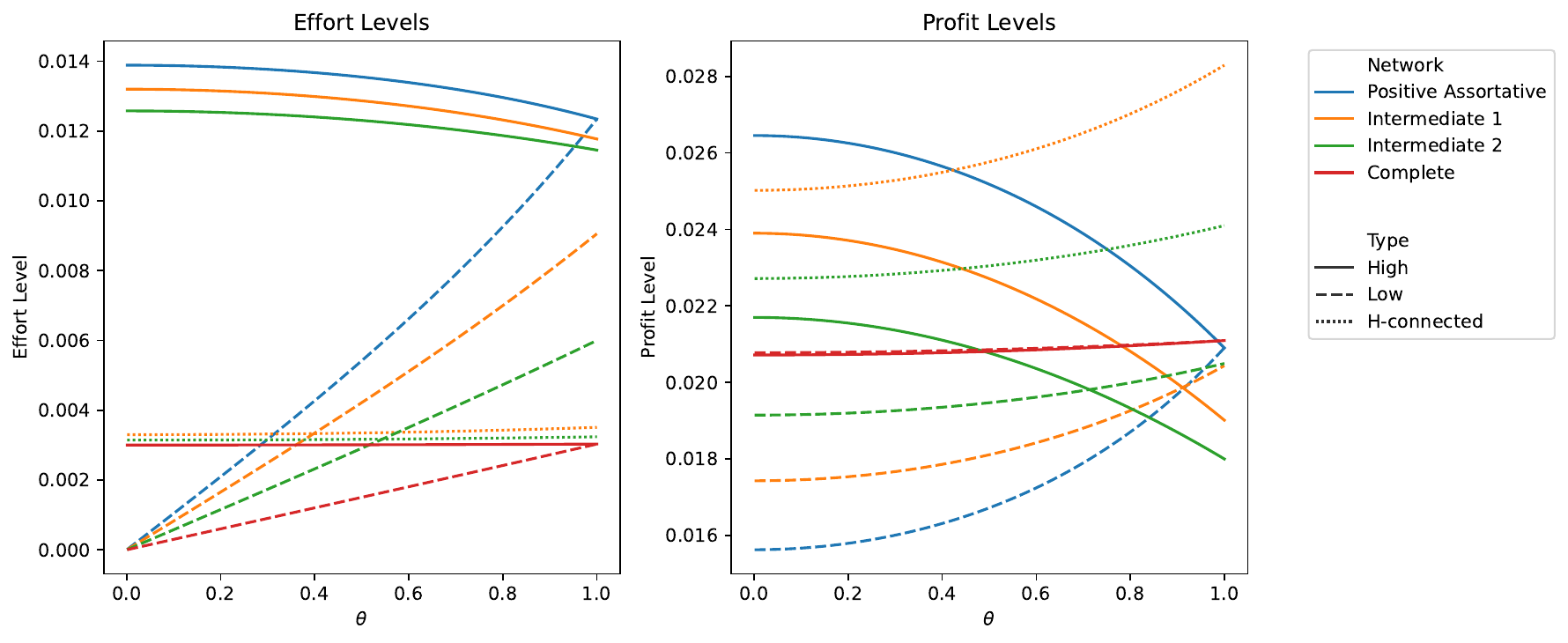}
    \end{subfigure}
    \fignote{
    This figure shows welfare (\textit{top}), individual firm efforts (\textit{bottom left}), and individual firm profits (\textit{bottom right}) for all network structures that are stable under $n = 6$ and $\rho = 0.5$. Colors represent different stable network structures showing a high-type firm connectivity-based transition from the PA network (blue) to the complete (red) network.
    All outcomes are plotted against the productivity of low-productivity firms ($\theta$). In the top panel, line style (dashed or solid) indicates stability, where only the structures with a non-zero stability region at $\phi = \underline{\phi}$ are included. In the lower panels, line style reflects firm type: solid for high-productivity firms, dashed for low-productivity firms, and dotted for the H-connected firms that are the high-productivity firms connected to all other firms in the intermediate structures.
    }
\end{figure}

The upper panel of Figure \ref{fig:n6_profit_effort_stability} shows four different stable network structures drawn in different colors. In these different networks, while low-type firms are always symmetric, high-type firms can have asymmetric connectivity with their own type. In particular, in the orange network, one high-type firm is connected to every other firm in the network; and in the green network, two high-type firms are connected to every other firm. In the figure, we label these high-type firms that become connected to the entire network as \textit{H-connected}. These intermediate configurations are pairwise stable for some $0<\theta<1$ values as shown in the upper panel and are relevant as they create a transition from the PA to the complete structure by increasing the connectivity of high-productivity firms. On the other hand, there are no stable networks that feature \textit{L-connected} firms ---those where a low-type firm connects to all other firms in the network--- and consequently, these are not drawn in the figure.

The welfare comparison reveals a clear ranking of network structures by their connectivity levels: the PA network delivers the highest welfare (stable for $\theta \leq \underline{\theta}$), while the complete network yields the lowest among the set of stable structures (stable for $\theta \geq \overline{\theta}$). However, this ranking holds only for intermediate values of $\theta$. When 
$\theta$ is sufficiently low or high, the PA or complete network, respectively, becomes the unique stable structure.

This leads to a counterintuitive insight: although welfare increases monotonically with $\theta$ for any given stable network, the level of $\theta$ determines which structure is pairwise stable, resulting in a non-monotonic welfare effect. As an example, starting from a low level of $\theta$ (i.e., R\&D productivity for low-type firms) where PA is the only pairwise stable configuration, welfare increases as $\theta$ rises up to a threshold $\bar{\theta}$. Further increases in $\theta$ past this threshold leads to a change in the stable structure from PA to complete, resulting in a discontinuous decrease in welfare that even falls below the starting point level. This observation has implications on how government and/or firm-level R\&D policies can have negative welfare effects even if they increase average R\&D productivity in the economy. Such non-monotonic effects of productivity are interesting, and further investigation on how R\&D policies endogenously affect welfare is left for future research. 

The lower panel of Figure \ref{fig:n6_profit_effort_stability} reports the efforts and profits of low-type, high-type, and H-connected type firms, labeled as dashed, solid, and dotted lines, respectively. The effort comparison shows that in intermediate configurations, high-type firms that only connect to other high-productivity firms have the highest effort level, whereas the H-connected firm has the lowest level of effort. On the other hand, in the PA and complete networks (where there is no H-connected type), it always holds that high-type firms have higher efforts than low-type firms. When profits are compared, in intermediate configurations, H-connected firms have the highest profits, but the ordering of low-type vs. high-type firms' profits depends on the $\theta$ level. In the PA network, it always holds that high-productivity firms have higher profits than low-productivity firms, with this ordering reversed for the complete network (as suggested by Corollary~\ref{corrolary_profit}) but with much smaller differences in profits. 

Finally, comparing the PA, complete, and other intermediate stable networks in Figure \ref{fig:n6_profit_effort_stability} reveals that group formation (under coalitional deviations) might serve as an alternative mechanism to study the network formation. While we employ the pairwise stability concept as commonly used in the R\&D networks literature, group formation analysis and its welfare implications are left as another avenue for future research.




\subsection{Crowding-out Effect}

Our next exercise experiments with $\rho$, the ratio of the high-productivity firms in the economy, to showcase a novel channel: a crowding-out effect of high-productivity firms on welfare. Intuitively, one might expect that when comparing two economies while holding $\theta$ fixed, the one with the higher number of high-productivity firms (i.e., with a higher $\rho$) would also have a higher welfare, as it has a larger average productivity. However, our results indicate that this relationship is mediated by the R\&D network. While welfare is monotonic in the fraction of high-type productivity firms for the complete network, it instead follows an inverted U-shape under the PA formation: both low and high fractions of high-productivity firms in the economy result in lower levels of welfare compared to an intermediate fraction of high-productivity firms. 

Under the PA formation, firms are clustered by type, which means that the average productivity in the economy is higher for a larger $\rho$ value, along with the size of the high-type connected component (with the low-type component shrinking in return) when we compare it with a setting with a lower $\rho$. This implies that an intermediate fraction of high-type firms in the economy maximizes welfare under the PA network formation process.

\begin{figure}[htbp]
    \centering
    \caption{Crowding-out effect of high-productivity firms on welfare}
    \label{fig:crowding_out}
    \includegraphics[width=\linewidth]{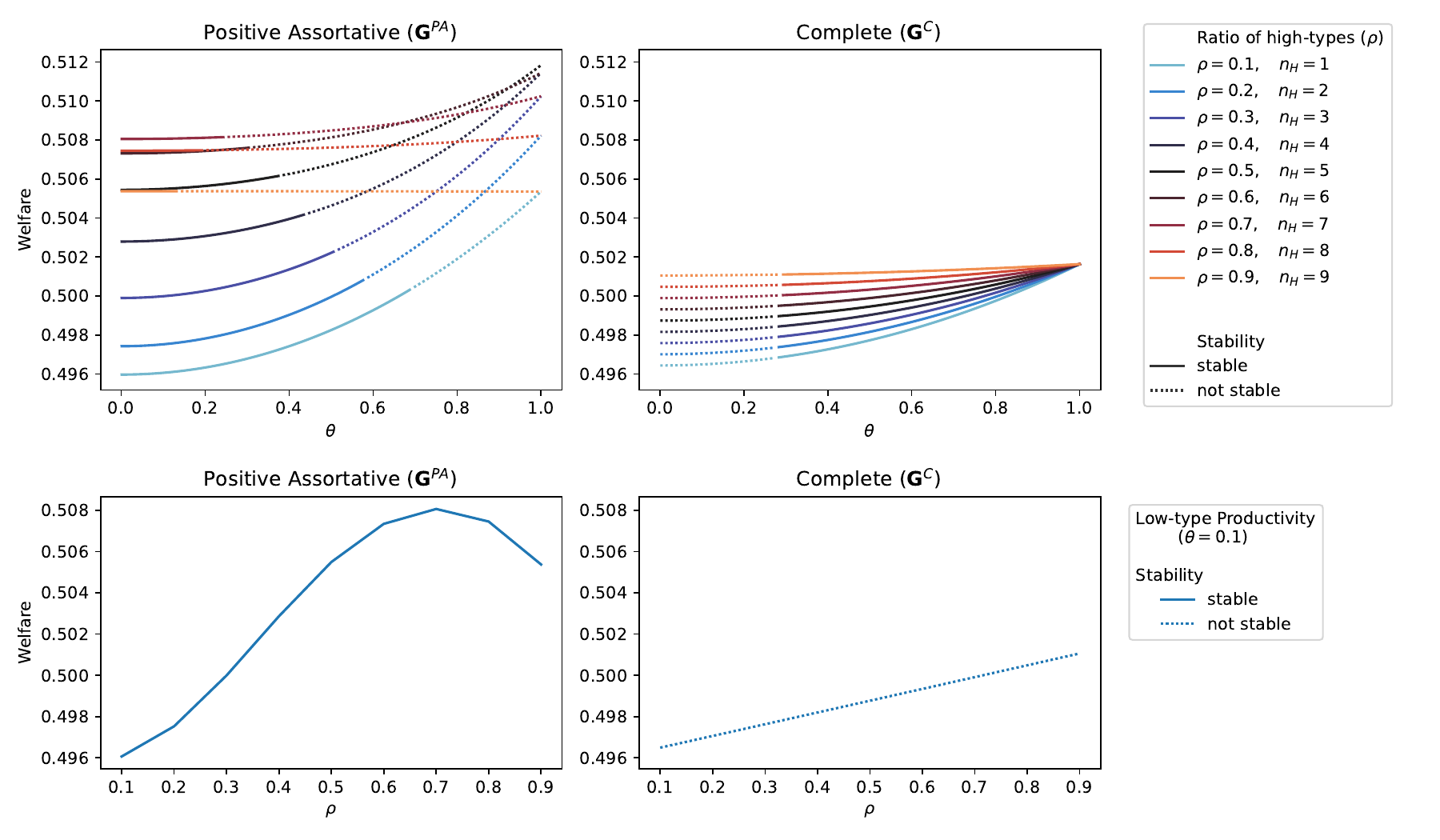}
    \fignote{
    This figure reports welfare for the positive assortative (\textit{left}) and complete (\textit{right}) networks when $n = 10$. The top row plots welfare over low-type productivity $\theta$, with colors indicating $\rho \in \{0.1,0.2,\dots,0.9\}$ and  (dashed) solid lines denoting (un)stability under $\phi = \underline{\phi}$. The top row shows that the PA network's stability region shrinks as $\rho$ increases, whereas the stability region of the complete network remains unchanged. The bottom row plots welfare with $\rho$ on the x-axis, for a single value of $\theta = 0.1$. For the PA network, welfare displays an inverted U-shape in $\rho$, capturing a crowding-out effect at higher values of $\rho$.
    }
\end{figure}

Figure \ref{fig:crowding_out} illustrates welfare levels for the PA and complete network structures as the fraction of high-productivity firms, $\rho$, increases. We modify the number of firms to $n = 10$ in our simulation results to provide a wider range of $\rho$ values. Under the PA structure, for an arbitrary low-type productivity $\theta$, welfare initially increases with the ratio of high-productivity firms but begins to decline once this ratio exceeds a certain threshold, as shown on the left side of the upper panel of Figure \ref{fig:crowding_out}. As the PA structure is stable when an economy exhibits large productivity gaps, shown in the lower panel of Figure \ref{fig:crowding_out} by setting $\theta=0.1$, the increasing ratio of high-productivity firms leads to an inverted U-shaped relationship with welfare in a network structure that remains stable.
In other words, when firms exhibit substantial heterogeneity in their productivities, there is a crowding-out effect of high-productivity firms on welfare in the stable PA configuration.

Figure \ref{fig:crowding_out} provides an additional insight: as $\rho$ rises, while the stability region of the complete network remains unchanged, this is not the case for the PA structure. Instead, the PA stability region shrinks as $\rho$ increases. Formally, the upper threshold $\overline{\theta}$, which ensures the stability of the PA network for any $\theta \leq \overline{\theta}$, decreases in the fraction of high-productivity firms, while $\underline{\theta}$, guaranteeing the stability of the complete network for any $\theta \geq \underline{\theta}$, remains unaffected.\footnote{For a comparison of how the stability regions of the PA and complete structures evolve under different settings, see Figure~\ref{fig:stable_PA_complete_n} in Appendix~\ref{apx:additional_results}.} 

As shown earlier in this Section (and illustrated in Figures \ref{fig:stability_2H2L} and \ref{fig:n6_profit_effort_stability}), the PA and complete network configurations are the main pairwise-stable structures under heterogeneity: together, they span almost the entire parameter space, meaning that when one structure is not stable, the other is typically stable. Therefore, understanding the crowding-out patterns in the PA structure becomes especially relevant in alternative economies where the PA structure is stable. 

Comparing the PA and complete networks across different values of $\rho$ reveals an important distinction. Under the complete formation, the connectivity structure remains fixed---every firm is linked to every other---and only the productivity distribution varies. Under the PA formation, by contrast, the connectivity structure itself evolves with $\rho$. Specifically, both the total number of links (i.e., network density) and the firms that constitute the connected components also change with $\rho$.\footnote{
Appendix Figure~\ref{fig:transiting_profit} disentangles this “composition-only” channel by holding the network fixed at a two-clique structure (two fully connected components of size $5$, as under PA when $\rho=1/2$) and varying only the productivity distribution via $\rho$. Upgrading firms one-by-one from $\theta$ to $1$ shows that the transitioning firm’s profit gain is larger for lower $\theta$ and is typically increasing in $\rho$, with a kink around $\rho=1/2$ when the fixed connectivity coincides with exact type clustering; the gain nevertheless remains positive. This exercise is illustrative rather than a stability result, since the imposed network is generally not pairwise stable away from $\rho=1/2$ (and sufficiently low $\theta$).
}

Specifically, the number of links in the PA structure is given by the sum of links within the two fully connected components---one with $n_H = \rho \cdot n$ high-productivity firms and the other with $n_L = (1-\rho)\cdot n$ low-productivity firms:
\begin{equation*}
    \text{Total Number of Links in PA} = \quad\frac{n_H\;(n_H - 1)}{2} + \frac{n_L\;(n_L - 1)}{2}
    = \frac{n\,(n-1)}{2} - n^2\, \rho\,(1-\rho).
\end{equation*}
This implies that as $\rho$ increases from 0 to 1, the number of links decreases until reaching a minimum at $\rho = 0.5$, symmetrically increasing thereafter. As \cite{goyal2001r} show, welfare in symmetric networks under the homogeneous setting exhibits an inverted-U relationship with the number of links. This naturally suggests that the crowding-out effect we observe in the PA structure under heterogeneity may stem from a generalized version of their result, where welfare displays an inverted-U shape in the network’s overall connectivity.

In the next section, we investigate this hypothesis by isolating the connectivity effect from the structural effect in the PA network. We find that although connectivity is correlated with the welfare pattern of the PA structure, it cannot fully account for the sharp increase at intermediate values or the decline at high values of $\rho$.

\subsection{Network Structure vs. Density Effect}

For symmetric networks with homogeneous firms, \cite{goyal2001r} show that welfare is concave in the average degree of the network, implying that welfare is maximized at an intermediate level of connectivity---a network configuration between the empty and complete. In our heterogeneous setting, we find a similar and more general connectivity effect, based on simulations with randomly generated networks instead of the more restricted symmetric network case. In Figure~\ref{fig:welfare_degree}, we depict the inverted U-shaped relationship between the network density and welfare for $n=10$ firms, where, for a given total number of links (on the $x$-axis), we average welfare over 1000 replications of networks where the links are randomly assigned.

\begin{figure}
    \centering
    \caption{Welfare over Network Link Density for Random vs. PA/Complete Structures}
    \label{fig:welfare_degree}
    \includegraphics[width=\linewidth]{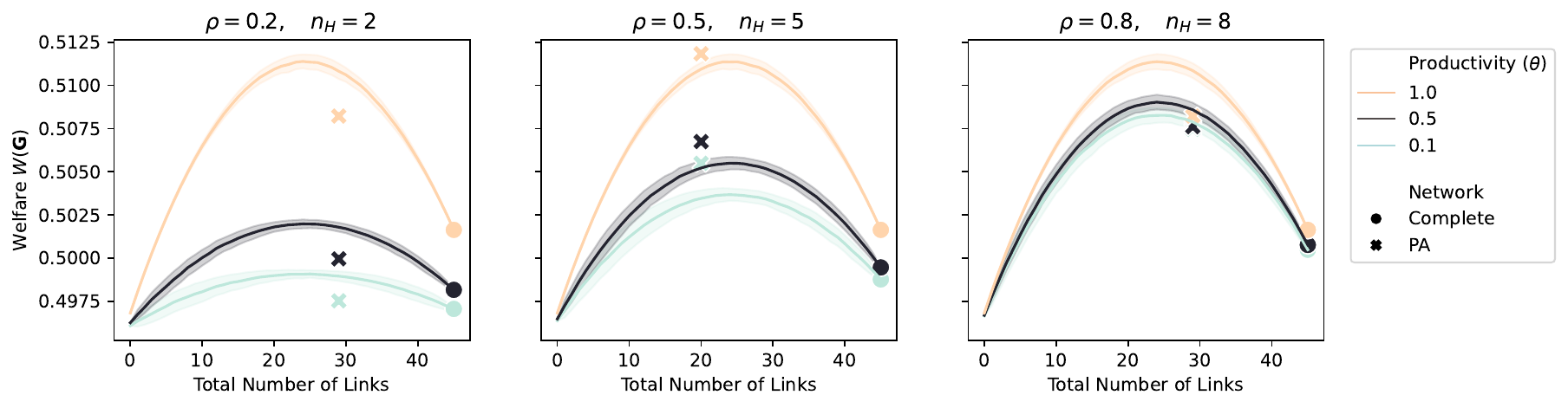}
    \fignote{
    Welfare is plotted against the total number of links in a randomly generated network for $n = 10$ under $\phi = \underline{\phi}$. Each panel corresponds to a different fraction of high-productivity firms in the economy, $\rho$: $0.2$ (\textit{left}), $0.5$ (\textit{middle}), and $0.8$ (\textit{right}). Colors represent different values of low-type productivity, $\theta \in \{0.1, 0.5, 1.0\}$. For each fixed total number of links (on the $x$-axis), networks are generated by randomly assigning links. Solid lines show average (across simulations) welfare, and shaded regions indicate plus and minus one standard deviation from the mean welfare.
    The welfare of the PA network and the complete network is marked with a cross and a circle, respectively.
    At the mid-level value of $\rho$, the PA network achieves a welfare level (for $\theta = 0.1$) that exceeds even the maximum average welfare obtained from the random networks when $\theta = 0.5$, highlighting the welfare advantage of the PA structure at intermediate values of $\rho$. This pattern does not occur for low or high values of $\rho$.
    }
\end{figure}

Furthermore, the simulations indicate that average welfare is positively correlated with the economy's average productivity, as the welfare curves shift upward as the productivity gap narrows (i.e., as $\theta$ increases) and as the proportion of high-productivity firms ($\rho$) rises.

However, when we compare the PA and complete structures to this average benchmark, a more complex picture emerges. While the complete network (circle markers in Figure \ref{fig:welfare_degree}) generally aligns with the shifts caused by $\theta$ and $\rho$ described in the previous paragraph, the PA structure (cross markers in Figure \ref{fig:welfare_degree}) exhibits distinct relationships with these quantities. To isolate this effect, Figure \ref{fig7_PA_random} provides a comparison by plotting welfare against the fraction of high-productivity firms ($\rho$) for both the PA structure and a set of randomly generated networks constrained to have the exact same total number of links.

\begin{figure}[htbp]
    \centering
    \caption{Network Structure Effect --- Welfare as a function of high-productivity firm ratio ($\rho$) in PA and Random Networks with Equal Total Number of Links}
    \includegraphics[width=\linewidth]{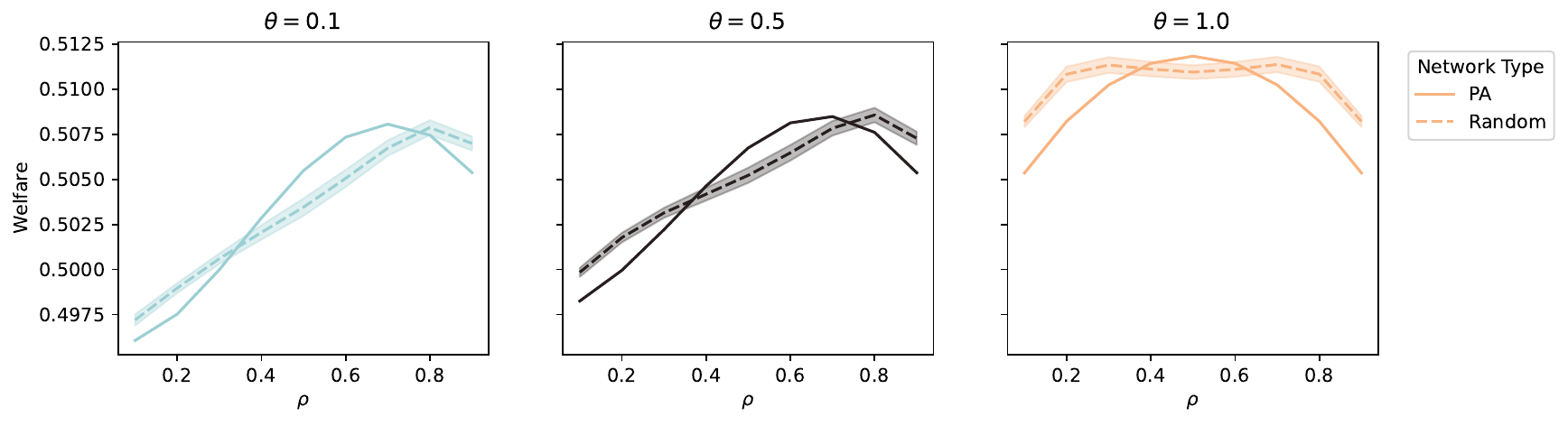}
    \label{fig7_PA_random}
    \fignote{
    Welfare is plotted for the positive assortative network alongside the average welfare from randomly generated networks with the same number of links for $n=10$ and $\phi = \underline{\phi}$. Panels correspond to different values of low-type productivity $\theta$: $0.1$ (\textit{left}), $0.5$ (\textit{middle}), and $1.0$ (\textit{right}). The solid lines show welfare under the PA structure. For each feasible $\rho$, the number of links in the corresponding PA network is computed, and random networks with the same $\rho$ and the same number of links are generated. The dashed lines denote the average welfare of these random networks, and the shaded regions represent plus and minus one standard deviation from the mean welfare.
    Although welfare in the PA and random networks are positively correlated, the differences between them are substantial, highlighting the effect of network structure on welfare.
    }
\end{figure}

Together, Figures \ref{fig:welfare_degree} and \ref{fig7_PA_random} indicate that the performance of the PA structure is highly non-monotonic relative to the average benchmark. At extreme values of $\rho$ (either low or high), the welfare in PA network is lower relative to the average welfare in random networks with equal link density. Conversely, at intermediate levels of $\rho$, the PA structure generates a significant welfare premium. This premium, however, diminishes as the productivity gap closes (i.e., as $\theta$ increases). 
The network structure advantage at intermediate $\rho$ is so pronounced in large productivity gap settings that it can override the effect of lower intrinsic productivity. For instance, at $\rho=0.5$, the PA network with a large productivity gap ($\theta=0.1$) achieves a welfare level exceeding even the maximum average welfare of random networks with a significantly smaller productivity gap ($\theta=0.5$).

Our findings highlight that, under heterogeneity, the specific topology of the PA network creates an additional effect that goes beyond connection density, such that simple measures like the total number of links cannot capture the full effect.

\section{Conclusion}
\label{sec:conclusion}

In this paper, we study R\&D network formation with heterogeneous productivities, generalizing the homogeneous framework in \cite{goyal2001r}. Firms differ in the potential benefit of their R\&D activities and alliances, creating heterogeneous gains from own R\&D efforts as well as from R\&D externalities in a network setting.

Under such heterogeneity, first, we show that when a low- and a high-productivity firm (having a symmetric position in the R\&D network) are connected, the relative gain of the less productive firm is higher than the more productive firm. Moreover, when such two firms are compared, the one with higher productivity exerts more R\&D effort, and hence pays a higher R\&D cost, than the one with lower productivity. These benchmark comparisons also imply that while low-productivity firms benefit from being connected to high-productivity firms, the reverse does not necessarily hold, creating heterogeneous incentives for link formation. We show the complete network---pairwise stable in a homogeneous setting---ceases to be stable when the productivity gap is sufficiently large. 

By using computer-assisted closed form solutions, we next describe the full set of stable network structures in a 4-firms setting. Our findings illustrate a stability feature: the positive assortative (PA) network structure, which is a productivity-level-based clustered network, and the complete network structure together almost cover the entire stability space. This means that when one of these two structures is not stable, the other typically is. Therefore, understanding the welfare and stability implications of the PA structure is crucial. This alternative configuration has unique implications for welfare measured as the aggregate consumer and producer surplus in the economy.

The PA structure yields higher welfare in comparison to the complete structure, uncovering an R\&D productivity puzzle: an economy with a higher average R\&D productivity may exhibit lower welfare due to (i) emergence of alternative stable structures and (ii) crowding-out of high-productivity firms. Finally, we compare the welfare in the PA and complete networks with the average welfare in random networks using simulations, isolating the effects of network structure on welfare from the network density effect. This provides further insights into how the productivity distribution, connectivity structure and network density determine aggregate welfare under heterogeneity. 

Altogether, our results highlight that R\&D productivity fundamentally shapes the organization of innovation, with distinct implications for firm-level outcomes and aggregate welfare. These findings suggest that policies aimed at enhancing firm-level R\&D productivities must account for their impact on endogenous R\&D alliances and effort, as these responses may reverse the intended welfare gains.


\bibliography{references}
\bibliographystyle{chicago}


\appendix

\section{Proofs}
    \label{apx:proofs}

We first provide the proofs of the propositions presented in the paper. Statements and proofs of Lemmas used in the main results are provided at the end of the section.

\begin{proof}[\textbf{Proof of Proposition \ref{prop:equilibrium_existance}}]	
    We first show that a bound on $\phi$ guarantees $\mathbf{A}(\mathbf{G})$ is non-singular, which establishes uniqueness of the effort profile $\str{\mathbf{e}}(\mathbf{G})$. We use the Levy-Desplanques theorem, which is a consequence of Gershgorin’s disc theorem \citep{gershgorin1931uber}: an $n \times n$ matrix $\mathbf{A}$ is nonsingular if every column is strictly diagonally dominant, i.e. $|A_{ii}| > \sum_{j\ne i} |A_{ij}|$ for all $i = 1, \ldots, n$.  For any firm $i \in \mathcal{N}$, we have
	\begin{align*}
		\sum_{j\ne i} |A_{ij}(\mathbf{G})|
		&= \sum_{j\ne i} \left|\theta_j\big\{(n+1)  G_{ij}-[1+ d_j(\mathbf{G})]\big\}\right|, \\
		&= \sum_{j\in \mathcal{N}_{i}(\mathbf{G})}\theta_j\big|n-d_j(\mathbf{G})\big|
		 + \sum_{k\in \mathcal{N}_{-i}(\mathbf{G})}\theta_k\big|1 + d_k(\mathbf{G})\big|.
	\end{align*}
	The cardinalities of the set of neighbours and non-neighbours are given by $|\mathcal{N}_{i}| = d_i(\mathbf{G})$ and $|\mathcal{N}_{-i}(\mathbf{G})| = n - 1 - d_i(\mathbf{G})$, respectively. For any neighbour $j\in \mathcal{N}_{i}(\mathbf{G})$, $1 \le d_j(\mathbf{G}) \le n-1$ so that $1 \le n - d_j(\mathbf{G}) \le n - 1$, and similarly for non-neighbours $k\in \mathcal{N}_{-i}(\mathbf{G})$, $0 \le d_k(\mathbf{G})\le n-2$ so that $1 \le 1 + d_k(\mathbf{G}) \le n - 1$. Additionally, recall $\theta_i \le 1$ for all $i \in \mathcal{N}$. Bounding quantities using their maximum for all firms in each of the sets gives
	\begin{equation*}
		s^* \coloneqq \max_{i \in \mathcal{N}} \sum_{j\ne i} |A_{ij}(\mathbf{G})| \le \max_{i \in \mathcal{N}} \{d_i(\mathbf{G}) \cdot (n - 1) + [n - 1 - d_i(\mathbf{G})](n - 1)\} \le (n - 1)^2.
	\end{equation*}
	A sufficient condition that guarantees $\mathbf{A}(\mathbf{G})$ is strictly diagonally dominant is therefore
	\begin{equation*}
		d^* \coloneqq \min_{i \in \mathcal{N}} |A_{ii}(\mathbf{G})| = \min_{i \in \mathcal{N}} \left| \frac{(n+1)^2\phi}{\theta_i [n-d_i(\mathbf{G})]}-\theta_i[n- d_i(\mathbf{G})]
		\right| > (n-1)^2 \ge s^*.
	\end{equation*}
	Using $\theta_i \le 1$ and $n - d_i(\mathbf{G}) \le n$, we see that any $\phi \ge 1$ ensures $A_{ii}(\mathbf{G}) > 0$. Under these bounds, $d^* \ge [(n+1)^2 \phi - n^2] / n $, and solving for $\phi$ yields the sufficient inequality
    \begin{equation*}
		\phi > \frac{n[(n-1)^2 + n]}{(n+1)^2}
		\eqqcolon \str{\phi}.
	\end{equation*}
    Any $\phi > \phi^* > n-1$ guarantees invertibility of $\mathbf{A}(\mathbf{G})$ and the fact we can express equilibrium efforts as $\mathbf{e}^*(\mathbf{G}) = (\alpha - \Bar{c}) \cdot \mathbf{A}(\mathbf{G})^{-1} \mathbf{1}_n$. We let $\mathbf{D}$ and $\mathbf{E}$ collect the diagonal and off-diagonal elements of $\mathbf{A}(\mathbf{G})$, respectively, such that for all $i,j \in \mathcal{N}$ with $i \ne j$ we have $D_{ii} = A_{ii}(\mathbf{G})$, $E_{ii} = 0$ and $E_{ij} = A_{ij}(\mathbf{G})$. We can then write $\mathbf{A}(\mathbf{G}) = \mathbf{D} + \mathbf{E} = (\mathbf{I}_n + \mathbf{E} \mathbf{D}^{-1}) \mathbf{D}$ and
    \begin{equation*}
        \mathbf{e}^*(\mathbf{G}) = (\alpha - \Bar{c}) \cdot \mathbf{D}^{-1} (\mathbf{I}_n + \mathbf{E} \mathbf{D}^{-1})^{-1} \mathbf{1}_n.
    \end{equation*}
    Note that invertibility of $\mathbf{D}$ follows as it is diagonal with $d^* > 0$ and that of $\mathbf{I}_n + \mathbf{E} \mathbf{D}^{-1}$ is guaranteed if $\|\mathbf{E} \mathbf{D}^{-1}\| < 1$ for any matrix norm $\| \cdot \|$ \citep[see, e.g., Corollary 5.6.16 of][]{Horn_Johnson_2012}. We will show this holds using the infinity norm defined as $\|\mathbf{A}\|_{\infty} \coloneqq \max_{1 \le i \le n} \sum_{j=1}^{m} |A_{ij}|$ for any $n \times m$ matrix $\mathbf{A}$. As $\alpha > \bar{c}$, equilibrium efforts $\mathbf{e}^*(\mathbf{G}) > 0$ follows if row sums $\mathbf{s} \coloneqq (\mathbf{I}_n + \mathbf{E} \mathbf{D}^{-1})^{-1} \mathbf{1}_n$ are entry-wise strictly positive.

    First, observe $\|\mathbf{D}^{-1}\|_{\infty} = \max_{i \in \mathcal{N}} |D_{ii}^{-1}| = 1/d^*$ and $\|E\|_{\infty} = s^*$. Under the uniqueness bound $\phi > \phi^*$, $\|\mathbf{E} \mathbf{D}^{-1}\|_{\infty} \le \|\mathbf{E}\|_{\infty} \|\mathbf{D}^{-1}\|_{\infty} = s^* / d^* < (n-1)^2 / (n-1)^2 = 1$, proving the invertibility of $\mathbf{I}_n + \mathbf{E} \mathbf{D}^{-1}$ and the fact we can write
    \begin{equation*}
        (\mathbf{I}_n + \mathbf{E} \mathbf{D}^{-1})^{-1} = \sum_{k=0}^{\infty} (-1)^{k} (\mathbf{E} \mathbf{D}^{-1})^k = \mathbf{I}_n + \sum_{k=1}^{\infty} (-1)^{k} (\mathbf{E} \mathbf{D}^{-1})^k.
    \end{equation*}
    Second, using this power series expansion, we can write $\mathbf{s} = \mathbf{1}_n + \bm{\varepsilon}$, where we define the remainder term $\bm{\varepsilon} \coloneqq \sum_{k=1}^{\infty} (-1)^{k} (\mathbf{E} \mathbf{D}^{-1})^k \mathbf{1}_n$. Showing $\|\bm{\varepsilon}\|_{\infty} = \max_{i \in \mathcal{N}} |\varepsilon_i| < 1$ is then enough to guarantee $s_i = 1 + \varepsilon_i > 0$ for all $i \in \mathcal{N}$. Observe that under the bound $\phi > \phi^*$ we have
    \begin{align*}
        \|\bm{\varepsilon}\|_{\infty} & \le \sum_{k=1}^{\infty} \|(\mathbf{E} \mathbf{D}^{-1})^k \|_{\infty} \cdot \|\mathbf{1}_n \|_{\infty} \le \sum_{k=1}^{\infty} \|\mathbf{E} \mathbf{D}^{-1} \|^k_{\infty} \le \sum_{k=1}^{\infty} \left( \frac{s^*}{d^*} \right)^k = \frac{s^*}{d^* - s^*}, \\
        & \implies \|\bm{\varepsilon}\|_{\infty} \le \frac{n(n-1)^2}{(n+1)^2\phi - n[n + (n-1)^2]}.
    \end{align*}
    Finally, note that any $\phi$ larger than the threshold $\underline{\phi}$ defined in \eqref{eq:phi_condition} ensures $\|\bm{\varepsilon}\|_{\infty} < 1$.
\end{proof}

\vspace*{-8mm}

\begin{proof}[\textbf{Proof of Proposition \ref{prop:summetric_ratio}}]

The FOC in equation~\eqref{eq:foc_effort_desc} for any firm $i \in \mathcal{N}$ can be written as
\begin{multline*}
	\left[\frac{\phi}{\theta_i \eta_i(\mathbf{G})} - \theta_i \eta_i(\mathbf{G})\right]\str{e}_i(\mathbf{G}) -\left[\eta_j(\mathbf{G}) - 1 + G_{ij}\right]  \theta_j \str{e}_j(\mathbf{G}) = \frac{\alpha - \bar{c}}{n+1} \\
    +\sum_{l\in\mathcal{N}_i(\mathbf{G}) \setminus \{j\}} \theta_l  \eta_l(\mathbf{G})  \str{e}_l(\mathbf{G}) - \sum_{k\in\mathcal{N}_{-i}(\mathbf{G}) \setminus \{j\}} [1 - \eta_k(\mathbf{G})] \theta_k \str{e}_k(\mathbf{G}).
\end{multline*}
For any two firms $i$ and $j$ having a symmetric position (Definition~\ref{definition_symmetry}), we have $\eta_i(\mathbf{G}) = \eta_j(\mathbf{G})$ as well as
\begin{equation*}
	\mathcal{N}_i(\mathbf{G}) \setminus \{j\} = \mathcal{N}_j(\mathbf{G}) \setminus \{i\} \quad \text{and} \quad \mathcal{N}_{-i}(\mathbf{G}) \setminus \{j\} = \mathcal{N}_{-j}(\mathbf{G}) \setminus \{i\}.
\end{equation*}
Therefore, the right-hand side of the FOCs for firms $i$ and $j$ with symmetric position are equal, from which we can obtain
\begin{multline*}
\left[\frac{\phi}{\theta_i \eta_i(\mathbf{G})} - \theta_i \eta_i(\mathbf{G})\right]\str{e}_i(\mathbf{G}) -\left[\eta_j(\mathbf{G}) - 1 + G_{ij}\right]  \theta_j \str{e}_j(\mathbf{G}) = \\\left[\frac{\phi}{\theta_j \eta_j(\mathbf{G})} - \theta_j \eta_j(\mathbf{G})\right]\str{e}_j(\mathbf{G}) -\left[\eta_i(\mathbf{G}) - 1 + G_{ji}\right]  \theta_i \str{e}_i(\mathbf{G}).
\end{multline*}
Rearranging the equation, we obtain the effort ratio between firms $i$ and $j$:
\begin{align*}
    \left[\frac{\phi}{\theta_i \eta_i(\mathbf{G})}-\theta_i (1-G_{ij})\right] \str{e}_i(\mathbf{G}) & = \left[\frac{\phi}{\theta_j \eta_j(\mathbf{G})}-\theta_j (1-G_{ji})\right] \str{e}_j(\mathbf{G}), \\
    \frac{\str{e}_i(\mathbf{G})}{\str{e}_j(\mathbf{G})} & = \frac{\theta_i}{\theta_j} \cdot \frac{\phi - \theta_j^2 \eta_j(\mathbf{G}) (1-G_{ij})}{\phi - \theta_i^2 \eta_i(\mathbf{G}) (1-G_{ji})}.
\end{align*}
Finally, substituting this into the equilibrium profit equation~\eqref{eq:eq_profit_effort}, we obtain the profit ratio:
\begin{equation*}
    \frac{\str{\pi}_i(\mathbf{G})}{\str{\pi}_j(\mathbf{G})} = \frac{\phi - \theta_i^2  \eta_i^2(\mathbf{G})}{\phi - \theta_j^2 \eta_j^2(\mathbf{G})}\cdot \left[ 
        \frac{\phi - \theta_j^2 \eta_j(\mathbf{G}) (1-G_{ij})}{\phi - \theta_i^2 \eta_i(\mathbf{G}) (1-G_{ji})}
        \right]^2.
\end{equation*}
\end{proof}

\vspace{-10mm}

\begin{proof}[\textbf{Proof of Corollary~\ref{corrolary_profit}}]
    Showing that $\str{\pi}_i(\mathbf{G}_{+ij}) < \str{\pi}_j(\mathbf{G}_{+ij})$ is straightforward, as we have:
    $$
    \frac{\str{\pi}_i(\mathbf{G}_{+ij})}{\str{\pi}_j(\mathbf{G}_{+ij})} = \frac{\phi - \theta_i^2 \eta^2_i(\mathbf{G}_{+ij})}{\phi - \theta_j^2 \eta^2_j(\mathbf{G}_{+ij})}.
    $$
    Due to the symmetric position assumption, $\eta^2_i(\mathbf{G}_{+ij}) = \eta^2_j(\mathbf{G}_{+ij})$, so the term $\phi - \theta_i^2 \eta^2_i(\mathbf{G}_{+ij}) > \phi - \theta_j^2 \eta^2_j(\mathbf{G}_{+ij})$ when $\theta_i > \theta_j$, implying $\str{\pi}_i(\mathbf{G}_{+ij})<\str{\pi}_j(\mathbf{G}_{+ij})$.
    Now, in the case when $G_{ij}=0$, we have
    \begin{align*}
        \frac{\str{\pi}_i(\mathbf{G}_{-ij})}{\str{\pi}_j(\mathbf{G}_{-ij})} 
        &= \frac{\phi - \theta_i^2 \eta^2_i(\mathbf{G}_{-ij})}{\phi - \theta_j^2 \eta^2_j(\mathbf{G}_{-ij})} \left(\frac{\phi - \theta_j^2 \eta_j(\mathbf{G}_{-ij})}{\phi - \theta_i^2 \eta_i(\mathbf{G}_{-ij})} \right)^2 \\
        &= \underbrace{\frac{\phi - \theta_i^2 \eta^2_i(\mathbf{G}_{-ij})}{\phi - \theta_i^2 \eta_i(\mathbf{G}_{-ij})} \cdot \frac{\phi - \theta_j^2 \eta_j(\mathbf{G}_{-ij})}{\phi - \theta_j^2 \eta^2_j(\mathbf{G}_{-ij})}}_{=f(\theta_i^2)/f(\theta_j^2)} \cdot \underbrace{\frac{\phi - \theta_j^2 \eta_j(\mathbf{G}_{-ij})}{\phi - \theta_i^2 \eta_i(\mathbf{G}_{-ij})}}_{>1},
    \end{align*}
    where we define $f(x):= [\phi - x \cdot \eta_i^2(\mathbf{G}_{-ij})] / \left[\phi - x \cdot \eta_i(\mathbf{G}_{-ij})\right]$.
    The last fraction is greater than 1, as $\theta_i>\theta_j$. We now show that $f(x)$ is an strictly increasing function, which makes the product of the first two terms $f(\theta_i^2)/f(\theta_j^2) > 1$, guaranteeing $\str{\pi}_i(\mathbf{G}_{-ij}) > \str{\pi}_i(\mathbf{G}_{-ij})$. Taking the derivative of $f$ with respect to $x$, we have:
    \begin{align*}
        \frac{d f(x)}{d x} 
        &= \frac{\phi \cdot \eta_i(\mathbf{G}_{-ij}) [1 - \eta_i(\mathbf{G}_{-ij})]}{[\phi - x \cdot \eta_i(\mathbf{G}_{-ij})]^2}.
    \end{align*}
    Since $\phi \ge 1$, $x \in (0, 1]$, and $\eta_i(\mathbf{G}_{-ij}) = \eta_j(\mathbf{G}_{-ij}) \in [2/(n+1), n/(n+1)] \subset (0, 1)$, we have $f'(x) > 0$.
\end{proof}

\vspace{-8mm}

\begin{proof}[\textbf{Proof of Proposition \ref{prop:FC_stable}}]

Without loss of generality, fix a productivity distribution $\{\theta_i\}_{i\in \mathcal{N}} \in (0, 1]^n$ such that $1 = \theta_1 \ge \cdots \ge \theta_n > 0$. Lemma \ref{lemma:fc_deviation_threshold} guarantees the existence of a set of thresholds $\{\theta_{ij}^*\}_{i \ne j \in \mathcal{N}}$ with $\theta_{ij}^* \in (0, \theta_i)$. On the one hand, notice that for each firm $j \in \mathcal{N} \setminus \{n\}$, all firms $i \in \mathcal{N}$ with $i > j$ would maintain the collaboration link $(i, j)$ since $\theta_{ij}^* < \theta_i \le \theta_j$. On the other hand, for each firm $j \in \mathcal{N} \setminus \{1\}$, all firms $i \in \mathcal{N}$ with $i < j$ would break the link $(i, j)$ from the complete network if $\theta_j < \theta_{ij}^* < \theta_i$. For each firm $j \in \mathcal{N} \setminus \{1\}$, we define $\Theta_j \coloneqq \{\theta_{ij}^* \mid i < j\}$ and $\theta_j^* \coloneqq \max \Theta_j$. We then define an economy-wide $\Theta \coloneqq \{ \theta_j^* \mid 1 < j \le n\}$. Note that no firm $i \ne 1$ wishes to break the link to the most productive firm 1 as $\theta_{i1}^* < \theta_i \le 1$.

First, let $\theta^* \coloneqq \min \Theta$. If there exists any firm $j \in \mathcal{N} \setminus \{1\}$ with productivity level $\theta_j < \theta^*$, then $\theta_j < \theta^* \le \theta_j^* = \theta_{ij}^*$ for some firm $i \in \mathcal{N}$ with $i < j$. Therefore, firm $i$ does not wish to maintain link $(i, j)$, leading to $\mathbf{G}^C$ not being pairwise stable. Second, let $\theta^{**} \coloneqq \max \Theta$. If all firms $j \in \mathcal{N} \setminus \{1\}$ satisfy $\theta_j \ge \theta^{**}$, then $\theta_j \ge \theta_j^* \ge \theta_{ij}^*$ for all firms $i \in \mathcal{N}$ with $i < j$. Therefore, all links are sustainable and the complete network remains pairwise stable. The proof is completed by noting $\theta^* \le \theta^{**}$ as they are the minimum and maximum, respectively, over the set $\Theta$.    
\end{proof}

\begin{lemma}
	\label{lemma:FC_complete}
	The complete network $\mathbf{G}^C$ is characterized by $G_{ij}^C = 1$ for all $i, j \in \mathcal{N}$ with $i \ne j$. The effort level of this configuration can be derived in closed form as
	\begin{equation}
		\str{e}_i(\mathbf{G}^C) =  \frac{(\alpha-\bar{c})\theta_i}{(n+1)^2 \phi - \sum_{j=1}^{n} \theta_j^2}. \label{eq:effort_complete}
	\end{equation}
\end{lemma}

\begin{proof}
	The FOC \eqref{eq:foc_effort_desc} for the complete network structure can be re-expressed as
	\begin{align*}
		\str{e}_i(\mathbf{G}^C) &= \frac{\theta_i}{(n+1)^2 \phi - \theta_i^2} \left[\alpha-\bar{c} + \sum_{j\ne i} \theta_j \str{e}_j(\mathbf{G}^C) \right] =  \frac{\theta_i}{(n+1)^2 \phi}\left[\alpha-\bar{c} + \sum_{j = 1}^{n} \theta_j \str{e}_j(\mathbf{G}^C) \right],
	\end{align*}
	where we use the fact that $\eta_i(\mathbf{G}^C) = [n - d_i(\mathbf{G}^C)]/(n + 1) = 1/(n+1)$ for all $i \in \mathcal{N}$. Pre-multiplying this equation by $\theta_j$ and summing across $j \in \mathcal{N}$ leads to
	\begin{equation*}
		\sum_{j=1}^{n} \theta_j \str{e}_j(\mathbf{G}^C) = \frac{(\alpha-\bar{c}) \sum_{j=1}^{n} \theta_j^2}{(n+1)^2 \phi - \sum_{j=1}^{n} \theta_j^2}.
	\end{equation*}
	Replacing this result into the FOC and re-arranging terms yields the closed-form solution for complete network efforts \eqref{eq:effort_complete}.
\end{proof}

\begin{lemma}
	\label{lemma:FC_deviate}
	Consider a network $\mathbf{G}^{C}_{-kl}$, which removes the link between arbitrary firms $k$ and $l$ while preserving all other connections from the complete network. The effort levels for all firms in this configuration can be derived in closed form as (for $i \in \mathcal{N} \setminus \{k, l\}$)
	\begin{align}
		\str{e}_k(\mathbf{G}^{C}_{-kl}) & = \frac{\alpha - \bar{c}}{b_k(1 - Q) - 2\theta_k Q + \theta_l \lambda [n(1-Q)-(1+Q)]} , \label{eq:efforts_k} \\
		\str{e}_l(\mathbf{G}^{C}_{-kl}) & = \frac{(\alpha - \bar{c}) \lambda}{b_k(1 - Q) - 2\theta_k Q + \theta_l \lambda [n(1-Q)-(1+Q)]} , \label{eq:efforts_l} \\
		\str{e}_i(\mathbf{G}^{C}_{-kl}) & = \frac{(\alpha - \bar{c}) \theta_i [b_k + 2\theta_k + (n+1) \theta_l \lambda]}{(n+1)^2 \phi \{b_k(1 - Q) - 2\theta_k Q + \theta_l \lambda [n(1-Q)-(1+Q)]\} }, \label{eq:efforts_i}
	\end{align}
	where we define
	\begin{align*}
        b_k \coloneqq \frac{(n+1)^2 \phi}{2\theta_k} - 2\theta_k, \quad
		\lambda \coloneqq \frac{\theta_l}{\theta_k} \cdot \frac{(n+1)\phi - 2\theta_k^2}{(n+1)\phi - 2\theta_l^2}, \quad \text{and} \quad
		Q \coloneqq \frac{\sum_{j \in \mathcal{N} \setminus \{k, l\}} \theta_j^2}{(n+1)^2 \phi}.
	\end{align*}
\end{lemma}

\begin{proof}
	Without loss of generality, we set $k = 1$ and $l = 2$ to consider the network $\mathbf{G}^{C}_{-12}$. In such case, we have $\eta_1(\mathbf{G}^{C}_{-12}) = \eta_2(\mathbf{G}^{C}_{-12}) = 2/(n+1)$ and $\eta_i(\mathbf{G}^{C}_{-12}) = 1/(n+1)$ for $i \ge 3$. We can then rewrite the system of first-order conditions \eqref{eq:foc_effort_desc} as
	\begin{align}
		b_1 \str{e}_1(\mathbf{G}^{C}_{-12}) &= (\alpha - \bar{c}) + \sum_{j=3}^{n} \theta_j \str{e}_j(\mathbf{G}^{C}_{-12}) - (n-1)\theta_2 \str{e}_2(\mathbf{G}^{C}_{-12}), \label{eq:1} \\
		b_2 \str{e}_2(\mathbf{G}^{C}_{-12}) &= (\alpha - \bar{c}) + \sum_{j=3}^{n} \theta_j \str{e}_j(\mathbf{G}^{C}_{-12}) - (n-1)\theta_1 \str{e}_1(\mathbf{G}^{C}_{-12}), \label{eq:2} \\
		b_i \str{e}_i(\mathbf{G}^{C}_{-12}) & = (\alpha - \bar{c}) + \sum_{j=3}^{n} \theta_j \str{e}_j(\mathbf{G}^{C}_{-12}) + 2\theta_1 \str{e}_1(\mathbf{G}^{C}_{-12}) + 2\theta_2 \str{e}_2(\mathbf{G}^{C}_{-12}) \label{eq:3}, \quad \text{for } i \ge 3
	\end{align}
	where we define the coefficients $b = (b_1, \ldots, b_n)$ as
	\begin{align*}
		b_1 = \frac{(n+1)^2 \phi}{2\theta_1} - 2\theta_1, \quad
		b_2 = \frac{(n+1)^2 \phi}{2\theta_2} - 2\theta_2, \quad \text{and} \quad
		b_i = \frac{(n+1)^2 \phi}{\theta_i}, \quad \text{for } i \ge 3.
	\end{align*}
	Equating the right-hand sides of \eqref{eq:1} to \eqref{eq:2}, we can obtain
	\begin{align}
		\str{e}_2(\mathbf{G}^{C}_{-12}) = \frac{b_1 - (n-1)\theta_1}{b_2 - (n-1)\theta_2} \str{e}_1(\mathbf{G}^{C}_{-12}) \implies
		\str{e}_2(\mathbf{G}^{C}_{-12}) = \lambda \str{e}_1(\mathbf{G}^{C}_{-12}), \label{eq:e2}
	\end{align}
	where we define $\lambda \coloneqq [b_1 - (n-1)\theta_1]/[b_2 - (n-1)\theta_2] = (\theta_2/\theta_1) \cdot [(n+1)\phi - 2\theta_1^2]/[(n+1)\phi - 2\theta_2^2]$ to represent this ratio. Now, equate the right-hand sides of \eqref{eq:1} to \eqref{eq:3} and follow similar steps to obtain
	\begin{equation}
		\str{e}_i(\mathbf{G}^{C}_{-12}) = \frac{b_1 + 2\theta_1 + (n+1) \theta_2 \lambda}{b_i} \str{e}_1(\mathbf{G}^{C}_{-12}), \quad \text{for } i \ge 3. \label{eq:e_i}
	\end{equation}
	Define $T$ as the sum over effective efforts for unaffected firms, such that $T \coloneqq \sum_{j=3}^{n} \theta_j \str{e}_j(\mathbf{G}^{C}_{-12})$. From \eqref{eq:e_i} we can obtain
	\begin{align}
		T &= [b_1 + 2\theta_1 + (n+1) \theta_2 \lambda] \cdot Q \cdot \str{e}_1(\mathbf{G}^{C}_{-12}). \label{eq:T_e1}
	\end{align}
	where we define $Q \coloneqq \sum_{j=3}^{n} \theta_j/b_j = \sum_{j=3}^{n} \theta_j^2/[(n+1)^2 \phi]$. Replace \eqref{eq:e2} and \eqref{eq:T_e1} back into \eqref{eq:1} to obtain an expression for the effort of firm 1:
	\begin{align}
		b_1 \str{e}_1(\mathbf{G}^{C}_{-12}) &= (\alpha - \bar{c}) + T - (n-1)\theta_2 \lambda \str{e}_1(\mathbf{G}^{C}_{-12}), \notag \\
		\str{e}_1(\mathbf{G}^{C}_{-12}) &= \frac{\alpha - \bar{c}}{b_1(1 - Q) - 2\theta_1 Q + \theta_2 \lambda[n(1-Q)-(1+Q)]} \label{eq:efforts_1}.
	\end{align}
	Plugging \eqref{eq:efforts_1} into \eqref{eq:e2} and \eqref{eq:e_i}, we can obtain the efforts for all remaining firms: $\str{e}_2(\mathbf{G}^{C}_{-12})$ as in \eqref{eq:efforts_l} and $\str{e}_i(\mathbf{G}^{C}_{-12})$ for $i \ge 3$ as in \eqref{eq:efforts_i}. With some algebra, these solutions can be expressed in terms of fundamental model quantities as (for $i \ge 3$)
	\begin{align*}
		\str{e}_1(\mathbf{G}^{C}_{-12}) &= \frac{2 (\alpha - \bar{c}) \theta_1 [(n + 1)\phi - 2 \theta_2^2]}{(n + 1)\phi[(n + 1)^2\phi(1 - Q) - 4 (\theta_1^2 + \theta_2^2)] + 4 \theta_1^2\theta_2^2 [(n + 1) Q - n + 3]}, \\
		\str{e}_2(\mathbf{G}^{C}_{-12}) &= \frac{2 (\alpha - \bar{c}) \theta_2 [(n + 1)\phi - 2 \theta_1^2]}{(n + 1)\phi[(n + 1)^2\phi(1 - Q) - 4 (\theta_1^2 + \theta_2^2)] + 4 \theta_1^2\theta_2^2 [(n + 1) Q - n + 3]}, \\
		\str{e}_i(\mathbf{G}^{C}_{-12}) &= \frac{(\alpha - \bar{c}) \theta_i [(n + 1)^2 \phi^2 - 4 (\theta_1^2 + \theta_2^2)]}{(n + 1)^2\phi^2[(n + 1)^2\phi(1 - Q) - 4 (\theta_1^2 + \theta_2^2)] + 4(n + 1)\phi\theta_1^2\theta_2^2[(n + 1) Q - n + 3]}.
	\end{align*}
\end{proof}

\vspace{-5mm}

\begin{lemma}[Link deviation in the complete network]
    \label{lemma:fc_deviation_threshold}
    For any pair of firms $(i, j) \in \mathcal{N}$ in a complete network $\mathbf{G}^C$, there exists a threshold over firm $j$'s productivity $\theta_j$, denoted by $\str{\theta}_{ij}(\theta_1, \ldots, \theta_{j-1}, \theta_{j+1}, \ldots, \theta_n ) \in (0, \theta_i)$, such that:
    \begin{equation}
        \pi^*_{i}\big(\mathbf{G}^C_{-ij}\big) > \pi^*_{i}\big(\mathbf{G}^C\big) \iff \theta_j < \str{\theta}_{ij}(\theta_1, \ldots, \theta_{j-1}, \theta_{j+1}, \ldots, \theta_n ).
    \end{equation}
\end{lemma}

\vspace{-8mm}

\begin{proof}[\textbf{Proof of Lemma \ref{lemma:fc_deviation_threshold}}]
    Without loss of generality, consider again the situation where firm 1 is in consideration for sustaining or breaking a link with firm 2. Using the equilibrium profit expression \eqref{eq:eq_profit_effort} in combination with Lemmas \ref{lemma:FC_complete} and \ref{lemma:FC_deviate}, the profit for the firm 1 under networks $\mathbf{G}^C$ and $\mathbf{G}_{-12}^C$ is given by
    \begin{align*}
        \str{\pi}_1(\mathbf{G}^C) & = \left[\frac{\phi}{\theta_1^2} (n+1)^2 - 1\right] \phi \left[\frac{(\alpha - \bar{c})\theta_1}{(n+1)^2 \phi - \sum_{j=1}^{n} \theta_j^2}\right]^2, \\
        \str{\pi}_1(\mathbf{G}^C_{-12}) & = \left[\frac{\phi}{\theta_1^2}\left(\frac{n+1}{2}\right)^2 - 1\right] \phi \left\{\frac{\alpha - \bar{c}}{b_1(1 - Q) - 2\theta_1 Q + \theta_2 \lambda[n(1-Q)-(1+Q)]} \right\}^2.
    \end{align*}
	Define $R \coloneqq \str{\pi}_1(\mathbf{G}^C_{-12}) / \str{\pi}_1(\mathbf{G}^C)$ as ratio of profits for firm 1 between of the network with a deviation to the complete network. Using the definitions of $b_1$ and $Q$, we can then express
    \begin{equation*}
        R = \frac{(n+1)^2 \phi - 4\theta_1^2}{(n+1)^2 \phi - \theta_1^2} \left\{\frac
        {(n+1)^2 \phi(1 - Q) - \theta_1^2 - \theta_2^2}
        {(n+1)^2 \phi(1 - Q) - 4\theta_1^2 + 2\theta_1\theta_2 \lambda[(n+1)(1-Q)-2]}
        \right\}^2.
    \end{equation*}
    First, we show that $\lim_{\theta_2 \to 0} R > 1$, proving that there exists a profitable deviation from complete when $\theta_2 \approx 0$ regardless of the remaining productivity distribution $\{\theta_j\}_{j=3}^{n}$. With $\theta_2$ approaching 0, the ratio simplifies to
    \begin{align*}
    	\lim_{\theta_2 \to 0} R & = \frac{(n+1)^2 \phi - 4\theta_1^2}{(n+1)^2 \phi - \theta_1^2} \left\{\frac
    	{(n+1)^2 \phi(1 - Q) - \theta_1^2}
    	{(n+1)^2 \phi(1 - Q) - 4\theta_1^2}
    	\right\}^2, \\
    	&= \underbrace{\frac{(n+1)^2\phi - 4\theta_1^2}{(n+1)^2\phi(1 - Q) - 4\theta_1^2} \cdot \frac{(n+1)^2\phi(1-Q) - \theta_1^2}{(n+1)^2\phi - \theta_1^2}}_{=g(4)/g(1)} \cdot \underbrace{\frac{(n+1)^2\phi(1-Q) - \theta_1^2}{(n+1)^2\phi(1-Q) - 4\theta_1^2}}_{>1},
   	\end{align*}
   	where we define $g(x):= [(n+1)^2\phi - x \cdot \theta_1^2]/[(n+1)^2\phi(1-Q) - x \cdot \theta_1^2]$. We next show that $g(x)$ is a strictly increasing function, which makes the product of the first two terms $g(4)/g(1) > 1$, guaranteeing $\lim_{\theta_2 \to 0} R > 1$. Taking the derivative of $g$ with respect to $x$, and using $n \ge 3$, $\phi \ge 1$, $x \in [1, 4]$, and $Q > 0$, we have \begin{align*}
   		\frac{d g(x)}{d g} = \frac{\theta_1^2(n+1)^2\phi Q}{[(n+1)^2 \phi - x]^2} > 0.
   	\end{align*}
   	Second, we will show $\lim_{\theta_2 \to \theta_1} R < 1$, such that that the complete network is pairwise stable when $\theta_2 \approx \theta_1$ regardless of the remaining productivity distribution $\{\theta_j\}_{j=3}^{n}$. With $\theta_2$ approaching $\theta_1$, $\lambda$ approaches 1 and the ratio $R$ simplifies to
   	\begin{align*}
   		\lim_{\theta_2 \to \theta_1} R & = \underbrace{\frac{(n+1)^2 \phi - 4\theta_1^2}{(n+1)^2 \phi - \theta_1^2}}_{< 1} \Bigg\{\underbrace{\frac
   		{(n+1)^2 \phi(1 - Q) - 2\theta_1^2}
   		{(n+1)^2 \phi(1 - Q) - 2\theta_1^2 + 2\theta_1^2[(n+1)(1-Q)-3]}}_{\coloneqq \rho}
   		\Bigg\}^2.
   	\end{align*}
    We are left with showing $\rho < 1$. As $\theta_i \le 1$ for all $i \in \mathcal{N}$, $Q \le (n-2)/[(n+1)^2\phi]$, which implies $(n+1)^2 \phi (1 - Q) \ge (n+1)^2 \phi - n + 2 \ge 2 \ge 2 \theta_1^2$. It is then sufficient to show $(n+1)(1-Q) > 3$ to guarantee $\rho < 1$. Note that we have
    \begin{align*}
    	(n+1)(1-Q) \ge n+1-\frac{n-2}{(n+1)\phi} > n \ge 3.
    \end{align*}
    Finally, note that the bounds introduced on fundamental model quantities ensure the denominator of $R$ is bounded away from 0. As a rational function of $\theta_2$, $R$ is therefore continuous on $\theta_2$, with $\lim_{\theta_2 \to 0} R > 1$ and $\lim_{\theta_2 \to \theta_1} R < 1$. The derivative of $R$ can be signed using the same bounds as in previous inequalities:
    \begin{align*}
    	\frac{\partial R}{\partial \theta_2} &= 2 \cdot \frac{(n+1)^2 \phi - 4\theta_1^2}{(n+1)^2 \phi - \theta_1^2} \cdot \frac
    	{N}{D^3} \left( \frac{\partial N}{\partial \theta_2} \cdot D - N \cdot \frac{\partial D}{\partial \theta_2} \right), \quad \text{with} \\
        \frac{\partial N}{\partial \theta_2} \cdot D - N \cdot \frac{\partial D}{\partial \theta_2} &= -2 \left\{ \theta_2 D + N [(n+1)(1-Q) - 2] \frac{2(n+1) \phi \lambda}{(n+1)\phi - 2\theta_2^2} \right\} < 0,
    \end{align*}
    where $N \coloneqq (n+1)^2 \phi(1 - Q) - \theta_1^2 - \theta_2^2$, and $D \coloneqq (n+1)^2 \phi(1 - Q) - 4\theta_1^2 + 2\theta_1\theta_2 \lambda[(n+1)(1-Q)-2]$ are the numerator and denominator, respectively, in the relative efforts squared term of $R$. Monotonicity and the intermediate value theorem guarantee the existence of a unique threshold $0 < \str{\theta}_{12} < 1$ where $R = 1$, at which the most productive firm 1 is indifferent between keeping or removing the R\&D collaboration to firm 2. That is, the complete network is not pairwise stable for all $\theta_2 < \str{\theta}_{12}$, while 1 wants to maintain the link to 2 if $\theta_2 \ge \str{\theta}_{12}$.
\end{proof}

\section{Additional Simulation Results}
\label{apx:additional_results}

\begin{figure}[htbp]
    \centering
    \caption{Profit impact of a productivity upgrade under a fixed two-clique network}
    \label{fig:transiting_profit}
    \includegraphics[width=0.8\linewidth]{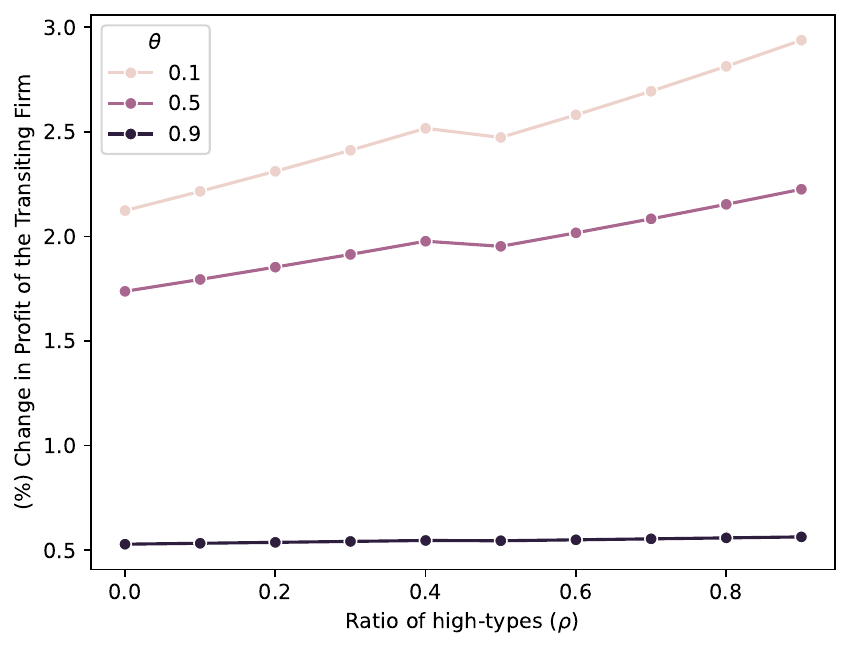}
    \fignote{This figure isolates the effect of changing the productivity composition while holding the network fixed. We consider $n=10$ firms partitioned into two fully connected components (two cliques of size $5$), mirroring the connectivity pattern of the Positive Assortative (PA) network at $\rho=\tfrac12$, and we set $\phi = \underline{\phi}$. Starting from $\rho=0$ (all firms have low productivity $\theta$), we increase $\rho$ in increments of $0.1$ by upgrading one firm at a time from $\theta$ to $1$ (high productivity). At each step, we record the change in profit for the \emph{transitioning} firm (the firm whose productivity is upgraded in that step), holding all other firms' productivities and the network links fixed. The x-axis reports $\rho$ (the share of high-productivity firms), and the y-axis reports the corresponding profit change of the transitioning firm. Curves are shown for three values of the low productivity parameter, $\theta\in\{0.1,0.5,0.9\}$. Two patterns emerge. First, the profit gain from upgrading is larger when $\theta$ is smaller, consistent with a larger productivity jump from $\theta$ to $1$. Second, the profit gain is generally increasing in $\rho$, except around $\rho=\tfrac12$, where the fixed network aligns with exact type clustering (all high types concentrated within one clique and all low types within the other); upgrading a low-type firm within the low-type clique yields a smaller marginal benefit than upgrading the remaining low-type firm in the high-type clique. In all cases, the profit change remains positive. Finally, note that this fixed-network exercise is not, in general, a stability analysis: the imposed two-clique network need not be pairwise stable for the corresponding productivity composition, except at $\rho=\tfrac12$ (and sufficiently low $\theta$).}
\end{figure}

\begin{figure}[htbp]
    \centering
    \caption{Stability Regions of PA and FC Networks in Large $n$ Settings}
    \label{fig:stable_PA_complete_n}
    \includegraphics[width=\textwidth]{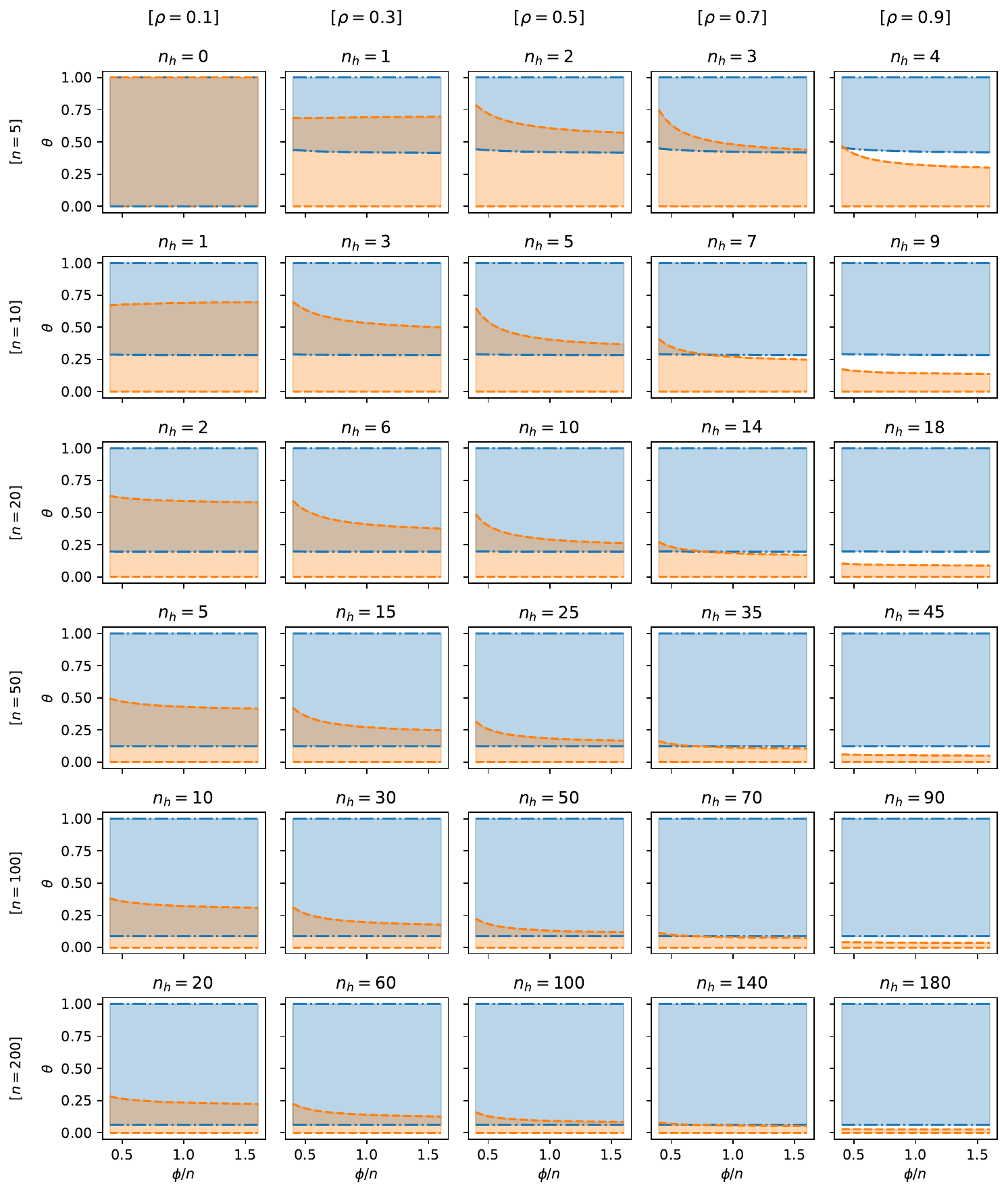}
    \fignote{Each row corresponds to a different number of firms, ranging from $n=5$ (top row) to $n=200$ (bottom row). Each column corresponds to a different share of high-productivity firms, $\rho$, increasing from $10\%$ (left) to $90\%$ (right). Orange regions indicate parameter combinations for which the PA network is stable, while blue regions indicate stability of the FC network.
    In light of the sufficient condition for the existence of an equilibrium effort profile (Proposition~\ref{prop:equilibrium_existance}), the cost coefficient $\phi$ on the horizontal axis is normalized by the number of firms. Specifically, since the lower bound satisfies $\underline{\phi} = \mathcal{O}(n)$, the horizontal axis is expressed in terms of $\phi/n$. As a result, within each row the horizontal axis spans the same range, while comparisons across rows account for differences in $n$ through this normalization.
    }
\end{figure}

\end{document}